\documentclass{llncs}

\usepackage[table]{xcolor}
\usepackage{wrapfig}
\usepackage{hhline}
\usepackage{multirow}
\usepackage{paralist}
\usepackage{tikz}
\usetikzlibrary{positioning,arrows}
\usepackage{multicol}
\usepackage{mathtools}

\usepackage{ltlfonts} %
\usepackage{xspace}
\usepackage{latexsym}
\usepackage{amsmath,amsfonts,amssymb}%,amsthm}
\newcommand{\EXPTIME}{{\sc Exptime}\xspace}

\newcommand{\PSPACE}{{\sc Pspace}\xspace}

\newcommand{\Nat}{{\mathbb{N}}}
\newcommand{\RealP}{{\mathbb{R}_+}}

\DeclareMathAlphabet{\mathpzc}{OT1}{pzc}{m}{it}

\newcommand{\Prop}{\mathcal{P}}

\newcommand{\tpl}[1]{(#1)}
\newcommand{\Lang}{{\mathcal{L}}}
\newcommand{\TLang}{{\mathcal{L}_T}}

\newcommand{\DTPA}{\text{\sffamily DTPA}}
\newcommand{\TA}{\text{\sffamily TA}}
\newcommand{\PDA}{\text{\sffamily PDA}}
\newcommand{\PTA}{\text{\sffamily PTA}}
\newcommand{\VPA}{\text{\sffamily VPA}}
\newcommand{\VPTA}{\text{\sffamily VPTA}}

\newcommand{\VPL}{\text{\sffamily VPL}}
\newcommand{\ECA}{\text{\sffamily ECA}}
\newcommand{\CARET}{\text{\sffamily CaRet}}
\newcommand{\ECNA}{\text{\sffamily ECNA}}
\newcommand{\CECNA}{\text{\sffamily C\_ECNA}}
\newcommand{\APCNA}{\text{\sffamily AP\_ECNA}}
\newcommand{\ARCNA}{\text{\sffamily AR\_ECNA}}
\newcommand{\ECVPA}{\text{\sffamily ECVPA}}

\newcommand{\Eventually}{\textsf{F}}

\newcommand{\Scall}{\Sigma_{\mathit{call}}}
\newcommand{\call}{{\mathit{call}}}
\newcommand{\ret}{{\mathit{ret}}}
\newcommand{\intA}{{\mathit{int}}}
\newcommand{\Sret}{\Sigma_{\mathit{ret}}}
\newcommand{\Sint}{\Sigma_{\mathit{int}}}
\newcommand{\Au}{\ensuremath{\mathcal{A}}}
\newcommand{\der}[1]{\ensuremath{\;\;{\mathop{{ %
            \longrightarrow}}\limits^{{#1}}}\!}\;\;} %

\newcommand{\abs}{\mathsf{a}}
\newcommand{\caller}{\mathsf{c}}
\newcommand{\Global}{\mathsf{g}}
\newcommand{\SUCC}{\mathsf{succ}}
\newcommand{\NULL}{\mathsf{\bot}}
\newcommand{\Pos}{{\mathit{Pos}}}
\newcommand{\val}{{\mathit{val}}}
\newcommand{\sval}{{\mathit{sval}}}
\newcommand{\Const}{{\mathit{Const}}}
\newcommand{\Intv}{{\mathit{Intv}}}
\newcommand{\Reg}{{\mathit{Reg}}}
\newcommand{\reg}{{\mathit{g}}}
\newcommand{\tw}{{\mathit{tw}}}
\newcommand{\Untimed}{{\mathit{Untimed}}}
\newcommand{\Timed}{{\mathit{Timed}}}
\newcommand{\dir}{\textit{dir}}

\newcommand{\Res}{\textit{Res}}

\newcommand{\con}{\textit{con}}
\newcommand{\first}{\textit{first}}
\newcommand{\live}{\textit{live}}
\newcommand{\push}{\textit{push}}
\newcommand{\pop}{\textit{pop}}
\newcommand{\bad}{\textit{bad}}
\newcommand{\new}{\textit{new}}
\newcommand{\Abs}{\textit{Abs}}
\newcommand{\AbsP}{\textit{AbsP}}

\newcommand{\MAP}{\textit{MAP}}
\newcommand{\infix}{\textit{infix}}
\newcommand{\TLangPred}{{\mathcal{L}_T^{\textit{pred}}}}
\newcommand{\TLangRec}{{\mathcal{L}_T^{\textit{rec}}}}
\newcommand{\TLangCaller}{{\mathcal{L}_T^{\textit{caller}}}}
\newcommand{\DefORmini}{\ensuremath{\;\big|\;}}
\newcommand{\true}{\ensuremath{\textup{\texttt{true}}}}
\begin{document}

\mainmatter              % start of the contributions
\title{Event-Clock Nested Automata}

\author{Laura Bozzelli  \and Aniello Murano  \and Adriano Peron}
\institute{Universit\`a degli Studi di Napoli Federico II, Italy }

\maketitle              % typeset the title of the contribution

\begin{abstract}
	In this paper we introduce and study Event-Clock Nested Automata (\ECNA), a formalism that combines Event Clock Automata (\ECA) and Visibly Pushdown Automata (\VPA). \ECNA\ allow to express real-time properties over non-regular patterns of recursive programs. We prove that \ECNA\ retain the %same 
closure and decidability properties of \ECA\ and \VPA\ being closed under Boolean operations and having a decidable language-inclusion problem. 
In particular, we prove that emptiness, universality, and language-inclusion for \ECNA\ are \EXPTIME-complete problems. As for the expressiveness, we have that \ECNA\ properly extend any previous attempt in the literature of combining \ECA\ and \VPA.
 \end{abstract}

\section{Introduction}

\emph{Model checking} is a well-established formal-method technique to automatically check for global correctness of reactive systems~\cite{Baier}.
%In this verification method, the (potentially infinite) dynamic behavior of a system, formally described by a mathematical model, is
%checked against a behavioral constraint expressed by a formal specification.
In this setting, automata theory over infinite words plays a crucial role: the set of possible (potentially infinite) behaviors of the system and the set of admissible
behaviors of the correctness specification can be modeled as languages accepted by automata. The verification problem of checking that a system meets its specification  then reduces
to testing language inclusion between two automata over infinite words.

In the last two decades,
%some attention has been devoted to %the analysis of infinite-state sequential systems modeled by
model checking of pushdown automata (\PDA) has received a lot of attention~\cite{BMP10,MP15,Wal96}.
PDA represent an infinite-state formalism suitable to model the control flow of typical sequential programs with
nested and recursive procedure calls. Although  the general problem of checking context-free properties of \PDA\ is undecidable~\cite{KPV02},
algorithmic solutions have been proposed for interesting subclasses of context-free requirements~\cite{AlurEM04,AlurMadhu04,AlurMadhu09,CMM+03}.
A well-known approach is that of \emph{Visibly Pushdown Automata} (\VPA)~\cite{AlurMadhu04,AlurMadhu09}, a subclass
of \PDA\ where the input symbols over a \emph{pushdown alphabet} control the admissible operations
on the stack. Precisely, the alphabet is partitioned into a set  of \emph{calls}, representing a
 procedure call and forcing a push stack-operation, a set
  of \emph{returns}, representing a procedure return and forcing a pop stack-operation, and a set of \emph{internal}
  actions that cannot access or modify the content of the stack.
%
%A call denotes  and the matching return (if any) along a given
%word denotes the exit from this procedure (i.e., a pop stack-operation). \VPA\ push
%onto the stack only when a call is read, pops the stack only at returns, and do not use the stack
%on reading internal symbols.
%
This restriction makes the class of
resulting languages  (\emph{visibly pushdown languages} or \VPL)
very similar in tractability and robustness to that of regular languages~\cite{AlurMadhu04,AlurMadhu09}.
%In particular,
\VPL\ are closed under Boolean operations, and language inclusion is \EXPTIME-complete.
\VPA\ capture all regular properties, and, additionally, allow to specify regular requirements over two kinds of \emph{non-regular}
patterns on input words: \emph{abstract paths} and \emph{caller paths}. An abstract path captures the local computation
within a procedure with the removal of subcomputations corresponding
 to nested procedure calls, while a caller path represents the call-stack content  at a given position of the input.

Recently, many works~\cite{AbdullaAS12,BenerecettiP16,BouajjaniER94,ClementeL15,EmmiM06,TrivediW10}  have investigated real-time extensions of \PDA\  by combining \PDA\ with
 \emph{Timed Automata} (\TA)~\cite{AlurD94}, a model widely used
to represent real-time systems. \TA\ are finite automata augmented with a finite set of real-valued clocks, which operate over words where each symbol is paired with a real-valued timestamp (\emph{timed words}).
%The clocks record the
%elapsed time among events, and while transitions are instantaneous, time can elapse in a control
%state.
All clocks progress at same speed and can
be reset by transitions (thus, each clock
keeps track of the elapsed time since the last reset). Constraints on clocks are associated with transitions to restrict the behavior of the automaton.
%The theory of \TA\ allows the verification of some verification problems for real-time systems~\cite{AlurD94}. In particular,
The emptiness problem for \TA\ is decidable and \PSPACE\ complete~\cite{AlurD94}.
However, since in \TA, clocks can be reset nondeterministically and independently of each other, the resulting class of timed languages is not closed under complement and, in particular,
language inclusion  is undecidable~\cite{AlurD94}. As a consequence, the
general verification problem (i.e., language inclusion) of formalisms combining unrestricted \TA\ with robust subclasses of \PDA\ such as \VPA\ is undecidable as well.
In fact,  checking language inclusion for \emph{Visibly Pushdown Timed Automata} (\VPTA) is undecidable even in the restricted case of specifications using at most one clock~\cite{EmmiM06}.
%(checking
%language inclusion between \TA\ which use at most one clock is instead decidable~\cite{OuaknineW04}).

\emph{Event-clock automata} (\ECA)~\cite{AlurFH99} are an interesting subclass of \TA\ where the explicit reset of clocks is disallowed. In \ECA, clocks have a predefined association with the input alphabet symbols. Precisely, for each symbol $a$ there are two clocks: the \emph{global recorder clock}, recording the time elapsed since the last occurrence of $a$, and the \emph{global predictor clock}, measuring the time elapsed since the next occurrence of $a$. Hence, the clock valuations are determined only by the input timed word being independent of the automaton behavior. Such a restriction makes the resulting class of timed languages closed under Boolean operations, and in particular, language inclusion  is  \PSPACE-complete~\cite{AlurFH99}. %Note that common real-time regular requirements like bounded-response time can be naturally modeled by \ECA.

Recently, a robust subclass of \VPTA, called \emph{Event-Clock Visibly Pushdown Automata} (\ECVPA), has been proposed in~\cite{TangO09}, combining \ECA\ with \VPA.  \ECVPA\ are closed under Boolean operations, and %their
language inclusion %probblem
is \EXPTIME-complete. However, \ECVPA\ do not take into account the nested hierarchical structure induced by a timed word over a pushdown alphabet, namely, they do not provide any explicit mechanism to relate the use of a stack with that of event clocks.
\vspace{-0.3cm}

\paragraph{Our contribution.}
In this paper, we introduce an extension of \ECVPA, called \emph{Event-Clock Nested Automata} (\ECNA) that, differently from \ECVPA, allows to relate the use of event clocks and the use of the stack. To this end, we add for each input symbol $a$ three additional event clocks: the \emph{abstract recorder clock} (resp., \emph{abstract predictor clock}), measuring the time elapsed since the last occurrence (resp., the time for the next occurrence)
of $a$ along the maximal abstract path  visiting the current position; the \emph{caller clock}, measuring the time elapsed
since the last occurrence of $a$ along the caller path from the current position. In this way, \ECNA\ allow to specify relevant real-time non-regular properties including:
\begin{compactitem}
\item Local bounded-time responses %properties
  such as  ``in the local computation  of a procedure $A$, every request $p$ is followed
  by a response $q$ within $k$ time units".
  \item Bounded-time total correctness requirements  such as ``if the pre-condition $p$ holds when the procedure $A$
  is invoked, then the procedure must return within $k$ time units and $q$ must hold upon return".
  \item Real-time security properties which require the inspection of the call-stack such as ``a module $A$ should be invoked only if module $B$ belongs to the call stack and within
  $k$ time units since the activation of module $B$".
\end{compactitem}
We show that \ECNA\ are strictly more expressive than \ECVPA\ and, as for \ECVPA, the resulting class of languages is closed under all Boolean operations.
Moreover, language inclusion and visibly model-checking of  \VPTA\ against \ECNA\ specifications are decidable and \EXPTIME-complete.
The key step in the proposed decision procedures is a translation
of \ECNA\ into equivalent \VPTA. %\ (recall that emptiness of \VPTA\ is decidable and \EXPTIME-complete \cite{}).
% Due to lack of space, many proofs are omitted  and can be found in the Appendix.\vspace{-0.2cm}

\paragraph{Related work.} Pushdown Timed Automata (\PTA) have been introduced in~\cite{BouajjaniER94}, %where it is shown (by exploiting the standard region-based technique~\cite{AlurD94}) that
and  their emptiness problem
is %decidable and
\EXPTIME-complete. An  extension of \PTA, namely Dense-Timed Pushdown Automata (\DTPA), has been studied in~\cite{AbdullaAS12},  where each symbol in the stack is equipped
with a real-valued
clock representing its `age' (the time elapsed since the symbol has been pushed onto the stack). It  has been shown in~\cite{ClementeL15} that \DTPA\ do not add expressive power and can be translated into equivalent \PTA. Our proposed translation of \ECNA\ into \VPTA\ is inspired from the construction in~\cite{ClementeL15}.  In~\cite{BhaveDKPT16}, an equally-expressive extension of \ECVPA~\cite{TangO09} over finite timed words,  by means of a \emph{timed} stack (like in \DTPA), is investigated. %A formalism more expressive than \PTA\ has been introduced independently in  ~\cite{} and ~\cite{}, where \PTA\ are extended with store/restore operations
%of clock valuations: at a matched call, the values of some clocks are stored on the stacks, and restored at the matching return. The emptiness problem for such a class of automata is in general undecidable. 
\section{Preliminaries}\label{sec:backgr}

In the following, $\Nat$ denotes the set of natural numbers and $\RealP$ the set of non-negative real numbers.
Let $w$ be a finite or infinite word over some alphabet. By $|w|$ we denote the length of $w$ (we set $|w|=\infty$ if $w$ is infinite). For all  $i,j\in\Nat $, with $i\leq j$, $w_i$ is
$i$-th letter of $w$, while $w[i,j]$ is the finite subword
%of $w$ given by
 $w_i\cdots w_j$.

An \emph{infinite timed word} $w$ over a finite alphabet $\Sigma$ is
an infinite word $w=(a_0,\tau_0) (a_1,\tau_1),\ldots$ over $\Sigma\times \RealP$ (intuitively, %for each $i$,
$\tau_i$ is the time at which $a_i$ occurs) such that the sequence $\tau= \tau_0,\tau_1,\ldots$ of timestamps  satisfies: (1) $\tau_i\leq \tau_{i+1}$ for all $i\geq 0$ (monotonicity), and (2) for all $t\in\RealP$, $\tau_i\geq t$ for some $i\geq 0$
(divergence). The timed word $w$ is also denoted by the pair $(\sigma,\tau)$, where $\sigma$ is the untimed word $a_0 a_1\ldots$
and $\tau$ is the sequence of timestamps.
An \emph{$\omega$-timed language} over $\Sigma$ is a set of infinite timed words over $\Sigma$.\vspace{-0.2cm}

\paragraph{Pushdown alphabets, abstract paths, and caller paths.}
A \emph{pushdown alphabet} is
a finite alphabet $\Sigma=\Scall\cup\Sret\cup\Sint$ which is partitioned into a set $\Scall$ of \emph{calls}, a set
$\Sret$ of \emph{returns}, and a set $\Sint$ of \emph{internal
  actions}. The pushdown alphabet $\Sigma$  induces a nested hierarchical structure in a given word over $\Sigma$ obtained by associating
to each call the corresponding matching return (if any) in a well-nested manner. Formally, the set  of \emph{well-matched words} is the set of finite words $\sigma_w$ over
$\Sigma$ inductively defined by the following grammar:
\[
\sigma_w:= \varepsilon \DefORmini
a\cdot \sigma_w \DefORmini
c\cdot \sigma_w \cdot r \cdot \sigma_w
\]
%
%\begin{compactitem}
%\item the empty word $\varepsilon$ is well-matched;
%\item if $a\in\Sint$ and $\sigma$ is well-matched, then  $a\cdot \sigma$ is well-matched;
%\item  if $c\in\Scall$,
%  $r\in\Sret$, and $\sigma,\sigma'$ are well-matched, then $c\cdot \sigma \cdot r \cdot \sigma'$ is well-matched.
%\end{compactitem}
where $\varepsilon$ is the empty word, $a\in\Sint$, $c\in \Scall$, and $r\in\Sret$.

Fix an infinite word $\sigma$ over $\Sigma$. For a call position $i\geq 0$,
if there is $j>i$ such that $j$ is a return position of $\sigma$ and
$\sigma[i+1,j-1]$ is a well-matched word (note that $j$ is
uniquely determined if it exists), we say that $j$ is the
\emph{matching return} of $i$ along $\sigma$.
For a position $i\geq 0$,
the \emph{abstract successor of $i$ along $\sigma$}, denoted
  $\SUCC(\abs,\sigma,i)$, is defined as follows:
  \begin{compactitem}
  \item If $i$ is a call,  then
    $\SUCC(\abs,\sigma,i)$ is the matching return of $i$ if such a matching return exists; otherwise
     $\SUCC(\abs,\sigma,i)=\NULL$ ($\NULL$ denotes the
    \emph{undefined} value).
  \item If $i$ is not a call, then $\SUCC(\abs,\sigma,i)=i+1$ if $i+1$ is
    not a return position, and $\SUCC(\abs,\sigma,i)=\NULL$, otherwise.
  \end{compactitem}
The \emph{caller of $i$ along $\sigma$}, denoted   $\SUCC(\caller,\sigma,i)$, is instead defined as follows:
\begin{compactitem}
  \item if there exists the greatest call position $j_c<i$ such that either   $\SUCC(\abs,\sigma,j_c)=\NULL$ or
  $\SUCC(\abs,\sigma,j_c)>i$, then  $\SUCC(\caller,\sigma,i)=j_c$; otherwise, $\SUCC(\caller,\sigma,i)=\NULL$.
  \end{compactitem}
A \emph{maximal abstract path} (\MAP) of $\sigma$ is a \emph{maximal} (finite or infinite) increasing sequence
of natural numbers $\nu= i_0<i_1<\ldots$ such that
$i_j=\SUCC(\abs,\sigma,i_{j-1})$ for all $1\leq j<|\nu|$. Note that for every position $i$ of $\sigma$, there is exactly one \MAP\ of $\sigma$ visiting
position $i$.  For each $i\geq 0$, the \emph{caller path of $\sigma$ from position $i$}, is the maximal
(finite) decreasing sequence of natural numbers $j_0>j_1\ldots >j_n$ such that $j_0= i$ and $j_{h+1}=\SUCC(\caller,\sigma,j_{h})$ for all $0\leq h <n$.
  %The abstract path $\nu$ is \emph{maximal} (both in the past and in the future)
%if there is no abstract path of $\sigma$ having $\nu$ as a proper subsequence.
Note that the positions of a \MAP\ have the same caller (if any). Intuitively, in the analysis of recursive programs, a maximal  abstract path
captures the local computation within a
procedure removing computation fragments corresponding to nested
calls, while the caller path represents the call-stack
content at a given position of the input.  %we denote by $\Pos(\abs,\sigma,i)$ the set of positions visited by the \MAP\ of $\sigma$ associated with position $i$.

For instance, consider the finite untimed word $\sigma_p$ of length $10$ depicted below %in Figure \ref{fig-word} over the pushdown alphabet $\Sigma$ with
where $\Scall =\{c\}$, $\Sret =\{r\}$, and $\Sint =\{\i\}$.
%\begin{wrapfigure}[6]{l}{0.5\linewidth}
\begin{figure}
\centering
\vspace{-0.6cm}
\begin{tikzpicture}[scale=1]
%\draw[draw=none,use as bounding box](-2.0,1.4) rectangle (2.0,-0.9);

\node (Word) at (-0.2,0.2) {};
\coordinate [label=left:{\footnotesize  $\sigma_p $\,\,$=$}] (Word) at (0.0,0.17);
\path[thin,black] (-0.1,0.2) edge   (10.1,0.2);

\node (NodeZero) at (0.0,0.0) {};
\coordinate [label=center:{\footnotesize  $0$}] (NodeZero) at (0.0,0.0);
\coordinate [label=below:{\footnotesize   \textbf{c}}] (NodeZero) at (0.0,-0.15);

\node (NodeOne) at (1.0,0.0) {};
\coordinate [label=center:{\footnotesize  $1$}] (NodeOne) at (1.0,0.0);
\coordinate [label=below:{\footnotesize   \textbf{c}}] (NodeOne) at (1.0,-0.15);

\node (NodeTwo) at (2.0,0.0) {};
\coordinate [label=center:{\footnotesize  $2$}] (NodeTwo) at (2.0,0.0);
\coordinate [label=below:{\footnotesize   \textbf{\i}}] (NodeTwo) at (2.0,-0.15);

\node (SourceOne) at (0.95,0.11) {};
\node (SourceSix) at (6.05,0.12) {};
\node (SourceSeven) at (7.0,0.10) {};
\node (SourceNine) at (9.0,0.10) {};
\node (SourceThree) at (3.0,0.10) {};
\node (SourceFive) at (5.0,0.10) {};

%\draw[->,thick,black] (SourceOne) .. controls (2.0,1.0) and  (5.0,1.0)  .. node[above] {\footnotesize $\SUCC(\abs,\sigma,1)$}  (SourceSix);
%\draw[->,thick,black] (SourceSeven) .. controls (7.5,1.0) and  (8.5,1.0)  .. node[above] {\footnotesize $\SUCC(\abs,\sigma,7)$}  (SourceNine);

\draw[->,thick,black] (SourceOne) .. controls (2.0,1.0) and  (5.0,1.0)  ..   (SourceSix);
\draw[->,thick,black] (SourceSeven) .. controls (7.5,1.0) and  (8.5,1.0)  ..   (SourceNine);

\node (NodeThree) at (3.0,0.0) {};
\coordinate [label=center:{\footnotesize  $3$}] (NodeThree) at (3.0,0.0);
\coordinate [label=below:{\footnotesize   \textbf{c}}] (NodeThree) at (3.0,-0.15);

\node (NodeFour) at (4.0,0.0) {};
\coordinate [label=center:{\footnotesize  $4$}] (NodeFour) at (4.0,0.0);
\coordinate [label=below:{\footnotesize   \textbf{\i}}] (NodeFour) at (4.0,-0.15);

\node (NodeFive) at (5.0,0.0) {};
\coordinate [label=center:{\footnotesize  $5$}] (NodeFive) at (5.0,0.0);
\coordinate [label=below:{\footnotesize   \textbf{r}}] (NodeFive) at (5.0,-0.15);

\draw[->,thick,black] (SourceThree) .. controls (3.5,0.6) and  (4.5,0.6)  .. node[above] {}  (SourceFive);

\node (NodeSix) at (6.0,0.0) {};
\coordinate [label=center:{\footnotesize  $6$}] (NodeSix) at (6.0,0.0);
\coordinate [label=below:{\footnotesize   \textbf{r}}] (NodeSix) at (6.0,-0.15);

\node (NodeSeven) at (7.0,0.0) {};
\coordinate [label=center:{\footnotesize  $7$}] (NodeSeven) at (7.0,0.0);
\coordinate [label=below:{\footnotesize   \textbf{c}}] (NodeSeven) at (7.0,-0.15);

\node (NodeEight) at (8.0,0.0) {};
\coordinate [label=center:{\footnotesize  $8$}] (NodeEight) at (8.0,0.0);
\coordinate [label=below:{\footnotesize   \textbf{\i}}] (NodeEight) at (8.0,-0.15);

\node (NodeNine) at (9.0,0.0) {};
\coordinate [label=center:{\footnotesize  $9$}] (NodeNine) at (9.0,0.0);
\coordinate [label=below:{\footnotesize   \textbf{r}}] (NodeNine) at (9.0,-0.15);

\node (NodeTen) at (7.0,0.0) {};
\coordinate [label=center:{\footnotesize  $10$}] (NodeTen) at (10.0,0.0);
\coordinate [label=below:{\footnotesize   \textbf{\i}}] (NodeTen) at (10.0,-0.15);

\end{tikzpicture}
%\caption{\label{fig-word} An untimed word over a pushdown alphabet}
\vspace{-0.7cm}
\end{figure}

\noindent Let $\sigma$ be $\sigma_p \cdot \i^{\omega}$. %, where $\sigma_p$ is the prefix  until position $10$ depicted in figure.
Note that $0$ is the unique unmatched call position of $\sigma$: hence, the \MAP\ visiting $0$ consists of just position $0$ and has no caller. The \MAP\ visiting position $1$ is the infinite sequence $1,6,7,9,10,11,12,13 \ldots$ and the associated caller is position $0$; %the \MAP\ visiting position $0$ consists of just position $0$ since it corresponds with a call devoid of a matching return ($succ(a,\sigma,0)=\NULL$);
the \MAP\ visiting position $2$ is the sequence $2,3,5$ and the associated caller is position $1$, and the \MAP\ visiting position $4$ consists of just position $4$ whose caller path is $4,3,1,0$.

\section{Event-clock nested automata}

In this section, we define the formalism of \emph{Event-Clock Nested Automata} (\ECNA), which allow a combined used of event clocks and visible operations on the stack. To this end, we augment the standard set of event clocks~\cite{AlurFH99} with a set of \emph{abstract event clocks} and a set of \emph{caller event clocks} whose values are determined by considering maximal abstract paths and caller paths of the given word, respectively. % which measure the time elapsed since the last occurrence (resp., the time to wait for the next occurrence) of the associated event along the maximal abstract path of the given timed word visiting the current position.

In the following, we fix a pushdown alphabet $\Sigma=\Scall\cup\Sret\cup\Sint$. The set $C_{\Sigma}$ of event clocks associated
with $\Sigma$ is given by $C_{\Sigma}:= \bigcup_{b\in\Sigma} \{x^{\Global}_b,y_b^{\Global},x_b^{\abs},y_b^{\abs},x_b^{\caller}\}$. Thus, we associate with each symbol $b\in \Sigma$, five event clocks:
the \emph{global recorder clock $x^{\Global}_b$} (resp., the \emph{global predictor clock $y^{\Global}_b$}) recording the time elapsed since the last occurrence of $b$ if any
(resp., the time required to the next occurrence of $b$ if any);  the \emph{abstract recorder clock $x_b^{\abs}$} (resp., the \emph{abstract predictor clock $y_b^{\abs}$}) recording the time elapsed since the last occurrence of $b$ if any (resp. the time required to the next occurrence of $b$) along the %in the %abstract path including \
\MAP\ visiting the current position; and the \emph{caller (recorder) clock} $x_b^{\caller}$ recording the time elapsed since the last occurrence of $b$ if any along the caller path from the current position.
Let $w=(\sigma,\tau)$ be an infinite timed word over $\Sigma$ and $i\geq 0$. We denote by  $\Pos(\abs,w,i)$ the set of positions visited by the \MAP\ of $\sigma$ associated with position $i$, and
by  $\Pos(\caller,w,i)$ the set of positions visited by the caller path of $\sigma$ from position $i$. In order to allow a uniform notation, we write $\Pos(\Global,w,i)$ to mean the full set $\Nat$ of positions.
The values of the clocks at a fixed position $i$ of the word $w$ can be deterministically determined as follows.

\begin{definition}[Determinisitic clock valuations] \emph{A \emph{clock valuation} over $C_{\Sigma}$ is a mapping $\val: C_{\Sigma} \mapsto \RealP\cup \{\NULL\}$, assigning to each event clock a value in
$\RealP\cup \{\NULL\}$ ($\NULL$ denotes the \emph{undefined} value).
Given an infinite timed word $w=(\sigma,\tau)$ over
 %the pushdown alphabet
 $\Sigma$ and a position $i$, the \emph{clock valuation $\val^{w}_i$ over $C_{\Sigma}$}, specifying the values of the event clocks at position $i$ along $w$, is defined as follows for each $b\in\Sigma$, where $\dir \in \{\Global,\abs, \caller\}$ and $\dir'\in \{\Global,\abs\}$:}
\[
\begin{array}{l}
%\val^{w}_i(x_b)  =  \left\{
%\begin{array}{ll}
%\tau_i - \tau_j
%&    \text{ if } \exists j<i:\, b= \sigma_j,  \text{ and }\, \forall k: j < k < i \Rightarrow b\neq  \sigma_k
%\\
%\NULL
%&    \text{ otherwise }
%\end{array}
%\right. \vspace{0.2cm}\\
%\val^{w}_i(y_b)  =  \left\{
%\begin{array}{ll}
%\tau_j - \tau_i
%&    \text{ if } \exists j>i:\, b= \sigma_j,  \text{ and }\, \forall k: i < k < j \Rightarrow b\neq \sigma_k
%\\
%\NULL
%&    \text{ otherwise }
%\end{array}
%\right. \vspace{0.2cm}\\
\val^{w}_i(x_b^{\dir})  =  \left\{
\begin{array}{ll}
\tau_i - \tau_j
&    \text{ if } \exists j<i:\, b= \sigma_j,\, j\in\Pos(\dir,w,i),  \text{ and}\,
\\
& \,\,\,\,\,\,\forall k: (j < k < i   \text{ and }\, k\in\Pos(\dir,w,i))  \Rightarrow b\neq \sigma_k\\
\NULL
&    \text{ otherwise }
\end{array}
\right.
\vspace{0.2cm}\\
\val^{w}_i(y_b^{\dir'})  =  \left\{
\begin{array}{ll}
\tau_j - \tau_i
&    \text{ if } \exists j>i:\, b= \sigma_j,\, j\in\Pos(\dir',w,i),  \text{ and}\,
\\
& \,\,\,\,\,\,\forall k: (i < k < j   \text{ and }\, k\in\Pos(\dir',w,i))  \Rightarrow b\neq \sigma_k\\
\NULL
&    \text{ otherwise }
\end{array}
\right.
\end{array}
\]
\end{definition}

%Thus, at position $i$ of the given  word, the value of the global recorder clock $x_b$ is the time since the last occurrence of symbol $b$
%if such an occurrence exists and it undefined otherwise; symmetrically,
%the value of the global predictor clock $y_b$ is the time to wait for  the next occurrence of $b$ if such an occurrence exists and it undefined otherwise.
It is worth noting that while the values of the global clocks are obtained by considering the full set of positions in $w$, the values of the abstract clocks (resp., caller clocks) are defined with respect to the \MAP\ visiting the current position (resp., with respect to the caller path from the current position). %   Notice that the values of the abstract recorder clock $x_b^{\abs}$ and the abstract predictor clock $y_b^{\abs}$ are not immediately defined in the timed word $w$ as the global clocks but with respect to the \MAP\ of $w$ visiting the current valuation position.

For  $C\subseteq C_{\Sigma}$ and a clock valuation
$\val$ over $C_{\Sigma}$, $\val_{\mid C}$ denotes the restriction of $\val$ to the set $C$.
 A \emph{clock constraint} over $C$ is a conjunction of atomic formulas of the form
$z \in I$, where $z\in C$, and
$I$ is either an interval in $\RealP$ with bounds in $\Nat\cup\{\infty\}$, or the singleton $\{\NULL\}$ (also denoted by $[\NULL,\NULL]$).
For a clock valuation $\val$ and a clock constraint $\theta$, $\val$ satisfies $\theta$, written
$\val\models \theta$, if for each conjunct $z\in I$ of $\theta$, $\val(z)\in I$. We denote by $\Phi(C)$ the set of clock constraints over $C$.

For technical convenience, we first introduce an extension of the known class of \emph{Visibly Pushdown Timed Automata} (\VPTA) \cite{BouajjaniER94,EmmiM06},
called \emph{nested} \VPTA. Nested \VPTA\ are simply \VPTA\ augmented with event clocks. Therefore, transitions of \emph{nested} \VPTA\ are constrained by a pair of disjoint finite sets of clocks: a finite set $C_{st}$ of standard clocks and a disjoint set $C\subseteq C_{\Sigma}$ of event clocks. As usual, a standard clock  can be reset when a transition is taken; hence, its value %valuation of a clock
at a position of an input word depends in general on the behaviour on the automaton and not only, as for event clocks, on the word.

The class of \emph{Event-Clock Nested Automata} (\ECNA) corresponds to the subclass of nested \VPTA\ where the set of standard clocks  $C_{st}$ is empty.

%Let $C_{st}$ be a finite set of standard clocks disjunct from $C_{\Sigma}$.
 A (standard) clock valuation over $C_{st}$ is a mapping $\sval: C_{st} \mapsto \RealP$ (note that the undefined value $\NULL$ is not admitted). For $t\in\RealP$ and a reset set $\Res\subseteq C_{st}$, $\sval+ t$ and $\sval[\Res]$ denote  the  valuations over $C_{st}$ defined as follows for all $z\in C_{st}$: $(\sval +t)(z) = \sval(z)+t$, and $\sval[\Res](z)=0$ if $z\in \Res$ and $\sval[\Res](z)=\sval(z)$ otherwise.
For $C\subseteq C_\Sigma$ and a valuation $\val$ over $C$, $\val\cup \sval$ denotes the valuation over $C_{st}\cup C$ defined in the obvious way.

\begin{definition}[Nested \VPTA]\emph{ A B\"{u}chi \emph{nested} \VPTA\   over  $\Sigma= \Scall\cup \Sint \cup \Sret$ is a tuple
$\Au=\tpl{\Sigma, Q,Q_{0},D=C\cup C_{st},\Gamma\cup\{\top\},\Delta,F}$, where $Q$ is a finite
set of (control) states, $Q_{0}\subseteq Q$ is a set of initial
states, $C\subseteq C_{\Sigma}$ is a set of event clocks, $C_{st}$ is a set of standard clocks disjunct from $C_{\Sigma}$, $\Gamma\cup\{\top\}$ is a finite stack alphabet,  $\top\notin\Gamma$ is the special \emph{stack bottom
  symbol}, $F\subseteq Q$ is a set of accepting states, and $\Delta_c\cup \Delta_r\cup \Delta_i$ is a transition relation, where ($D= C\cup C_{st}$):
  \begin{compactitem}
    \item $\Delta_c\subseteq Q\times \Scall \times \Phi(D) \times 2^{C_{st}} \times Q \times \Gamma$ is the set of \emph{push transitions},
    \item $\Delta_r\subseteq Q\times \Sret  \times \Phi(D) \times 2^{C_{st}} \times (\Gamma\cup \{\top\}) \times Q $ is the set of \emph{pop transitions},
       \item $\Delta_i\subseteq Q\times \Sint\times \Phi(D) \times 2^{C_{st}} \times Q $ is the set of \emph{internal transitions}.
  \end{compactitem}}
\end{definition}

 We  now describe how a nested \VPTA\ $\Au$ behaves over an infinite timed word $w$. Assume that on reading the $i$-th position of $w$, the current state of $\Au$ is $q$,  $\val^{w}_i$ is the event-clock valuation associated with  $w$ and %position
 $i$,  $\sval$ is the current valuation of the standard clocks in $C_{st}$, and $t=\tau_i-\tau_{i-1}$ is the time elapsed from the last transition (where $\tau_{-1}=0$). %(let in the following $\tau_{-1}=0$).
If $\Au$ reads a call $c\in \Scall$,  it chooses a push transition of the
form $(q,c,\theta,\Res,q',\gamma)\in\Delta_c$ and pushes the symbol $\gamma\neq \bot$ onto the
stack. If $\Au$ reads a return $r\in\Sret$,  it chooses a pop transition of the
form $(q,r,\theta,\Res,\gamma,q')\in\Delta_r$ such that $\gamma$ is the symbol on top of the stack, and
pops $\gamma$ from the stack (if $\gamma=\bot$, then $\gamma$ is read but not removed). Finally, on reading an internal action $a\in\Sint$, $\Au$ chooses an internal transition of the
form $(q,a,\theta,\Res, q')\in\Delta_i$, and, in this case, there is no operation on the stack. Moreover, in all the cases, the constraint $\theta$
of the chosen transition must be fulfilled  by the  valuation $(\sval+t)\cup (\val^{w}_i)_{\mid C}$, the control changes from $q$ to $q'$, and all the standard clocks in $\Res$ are reset (i.e., the  valuation of the standard clocks is updated to $(\sval +t)[\Res]$).

%If $\Au$ reads a call $c\in \Scall$,  it chooses a push transition of the
%form $(q,c,\theta,\Res,q',\gamma)\in\Delta_c$ such that the constraint $\theta$ is fulfilled  by the  valuation $(\sval+t)\cup (\val^{w}_i)_{\mid C}$, it pushes the symbol $\gamma\neq \bot$ onto the
%stack, resets the standard clocks in  $\Res$, and the control changes from $q$ to $q'$.
%If $\Au$ reads a return $r\in\Sret$,  it chooses a pop transition of the
%form $(q,r,\theta,\Res,\gamma,q')\in\Delta_r$ such taht $\gamma$ is the symbol on top of the stack and $\theta$ is fulfilled  by $(\sval+t)\cup (\val^{w}_i)_{\mid C}$, it pops $\gamma$ from the stack (if $\gamma=\bot$, then %$\gamma$ is read but not removed), it resets the standard clocks in $\Res$  and the control changes from $q$ to $q'$.
%Finally, on reading an internal action $a\in\Sint$, $\Au$ chooses an internal transition of the
%form $(q,a,\theta,\Res, q')\in\Delta_i$ such that the constraint $\theta$ is fulfilled  by $(\sval+t)\cup (\val^{w}_i)_{\mid C}$, it resets the standard clocks in $\Res$,
%and the control changes from $q$ to $q'$ without affecting the stack.

Formally, a configuration of $\Au$ is a triple $(q,\beta,\sval)$, where $q\in Q$,
$\beta\in\Gamma^*\cdot\{\top\}$ is a stack content, and $\sval$ is a valuation over $C_{st}$.
A run $\pi$ of $\Au$ over $w=(\sigma,\tau)$
%an infinite timed word $w=(\sigma,\tau)$ on $\Sigma$
is an infinite sequence  of configurations
 $\pi=(q_0,\beta_0,\sval_0),(q_1,\beta_1,\sval_1),\ldots
$ such that $q_0\in Q_{0}$,
$\beta_0=\top$, $\sval_0(z)=0$ for all $z\in C_{st}$ (initialization requirement), and the following holds
for all $i\geq 0$, where $t_i=\tau_i-\tau_{i-1}$ ($\tau_{-1}=0$):
\begin{itemize}
\item \textbf{Push:} If $\sigma_i \in \Scall$, then for some $(q_i,\sigma_i,\theta,\Res,q_{i+1},\gamma)\in\Delta_c$,
  $\beta_{i+1}=\gamma\cdot \beta_i$, $\sval_{i+1}= (\sval_i +t_i)[\Res]$,  and $(\sval_i +t_i)\cup(\val^{w}_i)_{\mid C}\models \theta$.
\item \textbf{Pop:} If $\sigma_i \in \Sret$, then for some
  $(q_i,\sigma_i,\theta,\Res,\gamma,q_{i+1})\in\Delta_r$, $\sval_{i+1}= (\sval_i +t_i)[\Res]$, $(\sval_i +t_i)\cup(\val^{w}_i)_{\mid C}\models \theta$,  and \emph{either}
  $\gamma\neq\bot$ and $\beta_{i}=\gamma\cdot \beta_{i+1}$, \emph{or}
  $\gamma=\beta_{i}= \beta_{i+1}=\top$.
\item \textbf{Internal:} If $\sigma_i \in \Sint$, then for some
  $(q_i,\sigma_i,\theta,\Res,q_{i+1})\in\Delta_i$, $\beta_{i+1}=\beta_i$, $\sval_{i+1}= (\sval_i +t_i)[\Res]$, and $(\sval_i +t_i)\cup(\val^{w}_i)_{\mid C}\models \theta$.
\end{itemize}
The run $\pi$ is \emph{accepting} if  there are infinitely many positions $i\geq 0$ such that $q_i\in F$.
The \emph{timed language} $\TLang(\Au)$ of $\Au$ is the set of infinite timed words $w$ over $\Sigma$
 such that there is an accepting run of $\Au$ on $w$. The \emph{greatest constant of $\Au$}, denoted $K_\Au$, is the greatest natural number used as bound in some clock constraint of $\Au$.
For technical convenience, we also consider nested \VPTA\ equipped with a \emph{generalized B\"{u}chi acceptance condition} $\mathcal{F}$ consisting of a family of sets of accepting states. In such a setting, a run $\pi$ is accepting if for each B\"{u}chi component $F\in \mathcal{F}$, the run $\pi$ visits infinitely often states in $F$.

  A \VPTA\ \cite{EmmiM06} corresponds to  a nested \VPTA\ whose set $C$ of event clocks is empty.
An \emph{\ECNA} is a  nested \VPTA\ whose set $C_{st}$ of standard clocks is empty. For
\ECNA, we can omit the reset component $\Res$ from the transition function and the valuation component $\sval$ from each configuration
$(q,\beta,\sval)$. Note the the class of Event-Clock Visibly Pushdown Automata  (\ECVPA)~\cite{TangO09}   corresponds to the subclass of  \ECNA\ where abstract and caller event-clocks are disallowed.
We also consider three additional subclasses of \ECNA: \emph{abstract predicting} \ECNA\ (\APCNA, for short) which do not use  abstract recorder clocks and caller clocks,
 \emph{abstract recording} \ECNA\ (\ARCNA, for short) which do not use  abstract predictor clocks and caller clocks, and \emph{caller} \ECNA\ (\CECNA, for short) which do not use abstract clocks. Note that these three subclasses of \ECNA\ subsume \ECVPA.

\begin{example}\label{example:ECNA} Let us consider the  \ARCNA\ %in Figure \ref{arcna}
depicted below, where $\Scall =\{c\}$, $\Sret =\{r\}$, and $\Sint=\{a,b,\i\}$. The control part of the transition relation ensures that
for each accepted word, the \MAP\ visiting the $b$-position associated with the transition $tr$ from $q_4$ to $q_5$ cannot visit the $a$-positions following the call positions. This implies that the  abstract recorder constraint
$x^{\abs}_{a}=1$ associated with $tr$ is fulfilled \emph{only if} all the %the visited
 occurrences of calls $c$ and returns $r$ are matched. %perfectly matched. % (i.e. the visited prefix of the word is well matched).
%\begin{wrapfigure}[6]{l}{0.5\linewidth}
\begin{figure}
\centering
\vspace{-0.5cm}
\begin{tikzpicture}[scale=1]
%\draw[draw=none,use as bounding box](-2.0,1.4) rectangle (2.0,-0.9);

\node[shape=circle,draw=black,inner sep=7pt,fill=white](InitNode) at (0.0,0.0) {};
\coordinate [label=center:{\footnotesize  $q_0$}] (LabelInit) at (0.0,0.0);

\node (Start) at (-1.0,0.0) {};
\path[->, thick,black] (Start) edge (InitNode);

\node[shape=circle,draw=black,inner sep=7pt,fill=white](NodeOne) at (1.5,0.0) {};
\coordinate [label=center:{\footnotesize  $q_1$}] (LabelOne) at (1.5,0.0);
\path[->, thick,black, loop above] (NodeOne) edge node[right] {\footnotesize \,  $c$, $\push(c)$}  (NodeOne);

\path[->, thick,black] (InitNode) edge node[above] {\footnotesize $a$} (NodeOne);

\node[shape=circle,draw=black,inner sep=7pt,fill=white](NodeTwo) at (4.0,0.0) {};
\coordinate [label=center:{\footnotesize  $q_2$}] (LabelTwo) at (4.0,0.0);
\path[->, thick,black, loop above] (NodeTwo) edge node[right] {\footnotesize \,  $a$}  (NodeTwo);

\path[->, thick,black] (NodeOne) edge node[above] {\footnotesize $c$, $\push(c)$ } (NodeTwo);

\node[shape=circle,draw=black,inner sep=7pt,fill=white](NodeThree) at (5.5,0.0) {};
\coordinate [label=center:{\footnotesize  $q_3$}] (LabelThree) at (5.5,0.0);
\path[->, thick,black, loop above] (NodeThree) edge node[right] {\footnotesize \,  $r$, $\pop(c)$}  (NodeThree);

\path[->, thick,black] (NodeTwo) edge node[above] {\footnotesize $a$} (NodeThree);

\node[shape=circle,draw=black,inner sep=7pt,fill=white](NodeFour) at (8.0,0.0) {};
\coordinate [label=center:{\footnotesize  $q_4$}] (LabelFour) at (8.0,0.0);
\path[->, thick,black, loop above] (NodeFour) edge node[right] {\footnotesize \,$b$}  (NodeFour);

\path[->, thick,black] (NodeThree) edge node[above] {\footnotesize $r$, $\pop(c)$ } (NodeFour);

\node[shape=circle,draw=black,inner sep=7pt,fill=white](NodeFive) at (10.5,0.0) {};
\node[shape=circle,draw=black,inner sep=6pt,fill=white](InternalCircleFive) at (10.5,0.0) {};
\coordinate [label=center:{\footnotesize  $q_5$}] (LabelFive) at (10.5,0.0);

\path[->, thick,black, loop above] (NodeFive) edge node[right] {\footnotesize\,$\i$}  (NodeFive);

\path[->, thick,black] (NodeFour) edge node[above] {\footnotesize $b$, $x_a^{\abs}=1$ } (NodeFive);

\end{tikzpicture}
%\caption{\label{arcna} Example of an abstract recording \ECNA.}
\vspace{-0.5cm}
%\end{wrapfigure}
\end{figure} 
Hence,  constraint $x^{\abs}_{a}=1$ %on the transition from $t$
ensures that the accepted language, denoted by $\TLangRec$,   consists of all the timed words of the form $(\sigma,\tau)\cdot (\i^{\omega},\tau')$ such that $\sigma$ is a well-matched word of the form $a \cdot c^{+} \cdot a^{+} \cdot r^{+} \cdot b^{+}$ and the time difference in $(\sigma,\tau)$  between the first and last symbols is $1$, i.e. $\tau_{|\sigma|-1}-\tau_0=1$. The example shows that $\ECNA$ allow  to express a meaningful real-time property of recursive systems, namely the ability of bounding the time required to perform an internal activity consisting of an unbounded number of returning recursive procedure calls.% (and returns).

Similarly, it is easy to define an  \APCNA\ accepting the timed language, denoted by $\TLangPred$, consisting of all the timed words of the form $(\sigma,\tau)\cdot (\i^{\omega},\tau')$ such that $\sigma$ is a well-matched word of the form $a^{+}  \cdot  c^{+} \cdot b^{+} \cdot r^{+} \cdot b$ and the time difference in $(\sigma,\tau)$  between the two extreme symbols is $1$.

Finally, as an example of language which can be defined by a \CECNA, we consider the timed language  $\TLangCaller$  consisting of the timed words of the form $(c,t_0)\cdot (\sigma,\tau)\cdot (\i^{\omega},\tau')$ such that $\sigma$ is a well-matched word of the form $a \cdot c^{+} \cdot a^{+} \cdot r^{+} \cdot b^{+}$ and the time difference in $(c,t_0)\cdot(\sigma,\tau)$  between the first and last symbols is $1$.
\end{example}

\paragraph{Closure properties of B\"{u}chi \ECNA.} As stated in the following theorem, the class of languages accepted by B\"{u}chi \ECNA\ is closed under Boolean operations. The proof exploits a technique similar to that used in~\cite{TangO09} to prove the analogous closure properties for \ECVPA\ (for details, see Appendix~\ref{APP:closureProperties}).

 \newcounter{theo-closureProperties}
\setcounter{theo-closureProperties}{\value{theorem}}

 \begin{theorem}\label{theorem:closureProperties} The class of $\omega$-timed languages accepted by B\"{u}chi
\ECNA\ is closed under union, intersection, and complementation. In particular, given two  B\"{u}chi
\ECNA\ $\Au=\tpl{\Sigma, Q,Q_{0},C,\Gamma\cup\{\top\},\Delta,F}$  and $\Au '=\tpl{\Sigma, Q',Q'_{0},C',\Gamma'\cup\{\top\},\Delta',F'}$ over $\Sigma_\Prop$, one can contruct
\begin{compactitem}
  \item a B\"{u}chi \ECNA\ accepting $\TLang(\Au)\cup \TLang(\Au')$ with $|Q|+|Q'|$ states, $|\Gamma|+|\Gamma'|+1$ stacks symbols, and greatest constant $\max(K_{\Au},K_{\Au'})$;
    \item a B\"{u}chi \ECNA\ accepting $\TLang(\Au)\cap \TLang(\Au')$ with $2|Q||Q'|$ states, $|\Gamma|
    |\Gamma'|$ stacks symbols,  and greatest constant $\max(K_{\Au},K_{\Au'})$;
      \item a B\"{u}chi \ECNA\ accepting the complement of $\TLang(\Au)$ with $2^{O(n^{2})}$ states,  $O(2^{O(n^{2})}\cdot |\Scall|\cdot |\Const|^{O(|\Sigma|)})$ stack symbols, and greatest constant
      $K_\Au$, where $n=|Q|$ and $\Const$ is the set of constants used in the clock constraints of $\Au$.
\end{compactitem}
\end{theorem}

\paragraph{Expressiveness results.} We now summarize the expressiveness results for \ECNA. First of all, the timed languages $\TLangRec$, $\TLangPred$, and $\TLangCaller$ considered in Example~\ref{example:ECNA} and definable by   \ARCNA, \APCNA, and \CECNA, respectively,  can be used to prove that the three subclasses \ARCNA,  \APCNA, and \CECNA\ of \ECNA\ are mutually incomparable. Hence, these subclasses strictly include the class of \ECVPA\ and are strictly included in \ECNA. % To this end, the basic property
The incomparability result directly follows  %stated
from Proposition~\ref{prop:expressivenessPredRec} below, whose proof %and proved
is in Appendix~\ref{APP:expressivenessPredRec}.

As for \ECNA, we have that they are less expressive than B\"{u}chi \VPTA. In fact, by Theorem~\ref{theorem:FromECNAtoVPTA} in Section~\ref{sec:DecisionProcedures}, B\"{u}chi \ECNA\ can be converted into equivalent B\"{u}chi \VPTA. %\ showing inclusion.
The inclusion is strict since, while  B\"{u}chi \ECNA\ are closed under complementation (Theorem~\ref{theorem:closureProperties}),   B\"{u}chi \VPTA\ are not~\cite{EmmiM06}.

%Finally, we can compared \ECNA\ with  \ECVPA.
In~\cite{BhaveDKPT16}, an equally-expressive extension of \ECVPA\ over finite timed words,  by means of a \emph{timed} stack, is investigated. The B\"{u}chi version of
	such an extension can be trivially encoded in  B\"{u}chi \ARCNA. Moreover, the proof of Proposition~\ref{prop:expressivenessPredRec} can also be used for showing that
	B\"{u}chi  \ECVPA\ with timed stack are less expressive than B\"{u}chi \ARCNA,  B\"{u}chi \APCNA, and B\"{u}chi \CECNA.
	
\noindent The general picture of the expressiveness results is summarized by Theorem \ref{th:express}.

 \newcounter{prop-expressivenessPredRec}
\setcounter{prop-expressivenessPredRec}{\value{proposition}}

\begin{proposition}\label{prop:expressivenessPredRec} The language $\TLangRec$ is \emph{not} definable by  B\"{u}chi \ECNA\ which do \emph{not} use abstract recorder clocks,  $\TLangPred$  is \emph{not} definable by B\"{u}chi \ECNA\ which do \emph{not} use abstract predictor clocks, and $\TLangCaller$ is \emph{not} definable by   B\"{u}chi \ECNA\ which do \emph{not} use caller clocks. %  an abstract-predicting (resp., abstract-recording) B\"{u}chi \ECNA.
Moreover, the language $\TLangRec\cup \TLangPred\cup \TLangCaller$ is \emph{not} definable by  B\"{u}chi \ARCNA,
   B\"{u}chi \APCNA\ and   B\"{u}chi \CECNA.
\end{proposition}

%By Theorem~\ref{theorem:FromECNAtoVPTA} in Section~\ref{sec:DecisionProcedures}, B\"{u}chi \ECNA\ can be converted into equivalent B\"{u}chi \VPTA. However, while  B\"{u}chi \ECNA\ are closed under complementation (Theorem~\ref{theorem:closureProperties}),   B\"{u}chi \VPTA\ are not~\cite{EmmiM06}. By Example~\ref{example:ECNA}, $\TLangRec$ is  definable by a B\"{u}chi  abstract-recording \ECNA\ (\ARCNA) and
%$\TLangPred$ is  definable by a B\"{u}chi  abstract-predicting \ECNA\ (\APCNA). Thus, by Proposition~\ref{prop:expressivenessPredRec}, we obtain the following relationships between the considered classes of automata (for brevity, in the following theorem, we omit the prefix B\"{u}chi). Note that the following expressiveness results also hold for the automata version over finite words.

%
%\begin{theorem}\label{th:express} The classes \ARCNA, \APCNA, and \CECNA\ are mutually incomparable, and $\APCNA \cup \ARCNA \cup \CECNA \subset  \ECNA $. Moreover, (i) $\ECVPA \subset \ARCNA$, (ii) $\ECVPA \subset \APCNA$, (iii) $\ECVPA \subset \CECNA$, and (iv) $\ECNA  \subset \VPTA$.
%\end{theorem}

\begin{theorem}\label{th:express} The classes \ARCNA, \APCNA, and \CECNA\ are mutually incomparable, and $\APCNA \cup \ARCNA \cup \CECNA \subset  \ECNA $. Moreover,  \vspace{-0.2cm}
\[
\begin{array}{ll}
  (1) \,\, \ECVPA \subset \ARCNA          & \quad  (2) \,\, \ECVPA \subset \APCNA   \hspace{5cm}\\
  (3) \,\, \ECVPA \subset \CECNA  & \quad (4)\,\, \ECNA   \subset \VPTA \hspace{5cm}
\end{array}
\]
\end{theorem}

Note that the expressiveness results above also hold for the automata version over finite timed words.

%\begin{remark} In~\cite{BhaveDKPT16}, an equally-expressive extension of \ECVPA\ over finite words,  by means of a \emph{timed} stack, is investigated. The B\"{u}chi version of
%such an extension can be trivially encoded in  B\"{u}chi \ARCNA. Moreover, the proof of Proposition~\ref{prop:expressivenessPredRec} can also be used for showing that
%B\"{u}chi  \ECVPA\ with timed stack are less expressive than both B\"{u}chi \ARCNA\ and B\"{u}chi \APCNA.
%\end{remark} 

\section{Decision procedures for B\"{u}chi \ECNA}\label{sec:DecisionProcedures}
%\vspace{-0.1cm}
In this section, we investigate the following decision problems:
\vspace{-0.2cm}
\begin{itemize}
  \item Emptiness,  universality, and language inclusion for B\"{u}chi \ECNA.
  \item \emph{Visibly model-checking problem against B\"{u}chi \ECNA:} given a \emph{visibly pushdown timed system $\mathcal{S}$} over $\Sigma$ (that is a B\"{u}chi \VPTA\ where all the states are accepting)
  and a B\"{u}chi \ECNA\ $\Au$ over  $\Sigma$,   does $\TLang(\mathcal{S})\subseteq \TLang(\Au)$ hold?
\end{itemize}
\vspace{-0.1cm}
We establish that the above problems are decidable and \EXPTIME-complete. The key intermediate result is an exponential-time translation
of B\"{u}chi \ECNA\ into language-equivalent generalized B\"{u}chi \VPTA.
 More precisely, we show that event clocks in \emph{nested} \VPTA\ can be removed with a single exponential blow-up.

\begin{theorem}[Removal of event clocks from nested \VPTA]\label{theorem:FromECNAtoVPTA}
Given a generalized B\"{u}chi \emph{nested} \VPTA\ $\Au$, one can construct in singly exponential time a generalized B\"{u}chi   \VPTA\ $\Au'$ (which do not use event clocks) such that
$\TLang(\Au')=\TLang(\Au)$ and $K_{\Au'}=K_{\Au}$. Moreover, $\Au'$ has $n\cdot 2^{O(p\cdot |\Sigma|)}$ states and $m + O(p)$ clocks, where $n$ is the number of $\Au$-states, $m$ is the number
of standard $\Au$-clocks, and $p$ is the number of \emph{event-clock} atomic constraints used by $\Au$.
\end{theorem}

In the following we sketch a proof of Theorem~\ref{theorem:FromECNAtoVPTA}. Basically, the result follows
%from a  generalized B\"{u}chi \emph{nested} \VPTA\ %$\Au$
%to  a generalized B\"{u}chi   \VPTA\ %$\Au'$
from a sequence of transformation steps all preserving language equivalence.
At each step, an event clock is replaced by a set of fresh standard  clocks.
To remove global event clocks we use the technique from~\cite{AlurFH99}.
Here, we focus on the removal of an \emph{abstract predictor} clock $y^{\abs}_b$ with $b\in \Sigma$, referring to Appendix~\ref{APP:RemovalAbstractRecorder} and~\ref{APP:RemovalCallerClocs} for the treatment of abstract recorder clocks and caller clocks.

Fix a generalized B\"{u}chi nested \VPTA\ $\Au=\tpl{\Sigma, Q,Q_{0},C\cup C_{st},\Gamma\cup\{\top\},\Delta,\mathcal{F}}$ such that $y^{\abs}_b \in C$.
By exploiting nondeterminism, we can assume that for each transition $\textit{tr}$ of $\Au$, there is exactly one atomic constraint $y^{\abs}_b\in I$ involving $y^{\abs}_b$ used as conjunct in the clock constraint
of $\textit{tr}$.  If $I\neq \{\NULL\}$, then   $y^{\abs}_b\in I$ is equivalent to a constraint of the form $y_b^{\abs}\succ \ell \wedge y_b^{\abs} \prec u$, where $\succ \in \{>,\geq\}$, $\prec\in \{<,\leq\}$,
$\ell \in \Nat$, and $u\in \Nat\cup\{\infty\}$. We call $y_b^{\abs}\succ \ell$ (resp., $y_b^{\abs} \prec u$) a lower-bound (resp., upper-bound) constraint. Note that if $u=\infty$ , the constraint $y_b^{\abs} \prec u$
is always fulfilled, but we include it to have a uniform notation.
We construct a generalized B\"{u}chi nested \VPTA\ $\Au_{y^{\abs}_b }$ equivalent to $\Au$ whose set of event clocks is  $C\setminus \{y^{\abs}_b\}$, and whose set of standard clocks is
$C_{st}\cup C_\new$, where $C_\new$ consists of the  fresh standard clocks $z_{\succ \ell}$ (resp., $z_{\prec u}$), for each lower-bound constraint  $y_b^{\abs}\succ \ell$  (resp., upper-bound constraint $y_b^{\abs} \prec u$)
of $\Au$ involving $y_b^{\abs}$.

We now report the basic ideas of the translation. Consider a lower-bound constraint $y_b^{\abs} \succ \ell$. Assume that a prediction $y_b^{\abs} \succ \ell$ is done by $\Au$
at position $i$ of the input word for the first time. Then, the simulating automaton $\Au_{y^{\abs}_b}$ exploits the standard clock $z_{\succ \ell}$
to check that the prediction
holds by resetting it at position $i$. Moreover, if $i$ is not a call (resp., $i$ is a call), $\Au_{y^{\abs}_b}$ carries the obligation $\succ$$\ell$ in its control state (resp., pushes the obligation $\succ$$\ell$ onto the stack)
in order to check that the constraint  $z_{\succ \ell}\succ \ell$ holds when the next $b$ occurs at a position $j_{\textit{check}}$ along the \MAP\ $\nu$ visiting position $i$.
We observe that:
\begin{compactitem}
  \item if a new prediction $y_b^{\abs} \succ \ell$ is done by $\Au$ at a position $j>i$ of   $\nu$ strictly preceding  $j_{\textit{check}}$, $\Au_{y^{\abs}_b}$ resets
  the clock $z_{\succ \ell}$ at position $j$ rewriting the old obligation. This is safe since the fulfillment of the lower-bound prediction $y_b^{\abs} \succ \ell$ at $j$  guarantees that prediction
  $y_b^{\abs} \succ \ell$ is fulfilled at $i$ along $\nu$.
  \item If a call position $i_c\geq i$ occurs in $\nu$ before $j_{\textit{check}}$, the next position of $i_c$  in $\nu$ is the matching return $i_r$ of $i_c$, and any \MAP\ visiting a position
  $h\in [i_c+1,i_r-1]$ is finite and ends at a position $k<i_r$. Thus, the clock $z_{\succ \ell}$ can be safely reset to check the prediction $y_b^{\abs} \succ \ell$ raised in  positions in $[i_c+1,i_r-1]$ since this check  ensures that $z_{\succ \ell}\succ \ell$ holds at position $j_{\textit{check}}$.%   guarantees the fulfillment of constraint also the prediction done at position $i$.
\end{compactitem}\vspace{0.1cm}

\noindent Thus, previous obligations on a constraint $y_b^{\abs} \succ \ell$ are always rewritten by more recent ones. At each position $i$, $\Au_{y^{\abs}_b}$ records in its control state the lower-bound obligations for the current \MAP\ $\nu$ (i.e., the \MAP\ visiting the current position $i$). Whenever  a call $i_c$ occurs, the lower-bound obligations are pushed on the stack in order to be recovered at the matching return $i_r$. If $i_c+1$ is not a return (i.e., $i_r\neq i_c+1$), then $\Au_{y^{\abs}_b}$ moves to a control state having an empty set of lower-bound obligations (position $i_c+1$ starts the \MAP\ visiting $i_c+1$).

The treatment of an upper-bound constraint $y_b^{\abs}  \prec u$ is symmetric. Whenever a prediction $y_b^{\abs}  \prec u$ is done by $\Au$ at a position $i$, and the simulating automaton $\Au_{y^{\abs}_b}$ has no obligation  on the constraint $y_b^{\abs}  \prec u$, $\Au_{y^{\abs}_b}$ resets the standard clock $z_{\prec u}$.  If $i$ is not a call (resp., $i$ is a call) the fresh obligation ($\first$,$\prec$$u$) is recorded in the control state (resp., ($\first$,$\prec$$u$) is pushed  onto the stack). When, along the \MAP\ $\nu$ visiting position $i$, the next $b$ occurs at a position $j_{\textit{check}}$,   the constraint  $z_{\prec u}\prec u$ is checked, and the obligation ($\first$,$\prec$$u$) is removed or confirmed (in the latter case, resetting the clock $z_{\prec u}$), depending on whether the prediction $y_b^{\abs}  \prec u$ is asserted at position  $j_{\textit{check}}$ or not. We observe that:
\begin{compactitem}
  \item if a new prediction $y_b^{\abs} \prec u$ occurs in a position $j>i$ of   $\nu$ strictly preceding  $j_{\textit{check}}$, $\Au_{y^{\abs}_b}$ simply ignores it (the clock $z_{\prec u}$  is not reset at position $j$) since checking the  prediction $y_b^{\abs} \prec u$ at the previous position $i$ guarantees the fulfillment of the
  prediction  $y_b^{\abs} \prec u$ at the position $j>i$ along $\nu$.
  \item If a call position $i_c\geq i$ occurs in $\nu$ before $j_{\textit{check}}$, then all the predictions $y_b^{\abs}  \prec u$ occurring in a \MAP\ visiting a position
  $h\in [i_c+1,i_r-1]$, with $i_r\leq j_{\textit{check}}$ being the matching-return of $i_c$, can be safely ignored (i.e., $z_{\prec u}$ is not reset there) since they are subsumed by the prediction at position $i$.
\end{compactitem}\vspace{0.1cm}

\noindent Thus, for new obligations on an upper-bound constraint $y_b^{\abs} \prec u$, the clock $z_{\prec u}$ is not reset. Whenever a
call $i_c$ occurs, the updated set $O$ of  upper-bound and lower-bound obligations   is pushed onto the stack  to be recovered at the matching return $i_r$ of $i_c$.
Moreover, if $i_c+1$ is not a return (i.e., $i_r\neq i_c+1$), then $\Au_{y^{\abs}_b}$ moves to a control state where the set of lower-bound obligations is empty and the set  of upper-bound  obligations is
obtained from $O$ by replacing each upper-bound obligation ($f$,$\prec$$u$), for $f\in \{\live,\first\}$, with the live obligation ($\live$,$\prec$$u$).
The latter % A live obligation ($\live$,$\prec$$u$)
asserted at the initial position
$i_c+1$ of the \MAP\  $\nu$ visiting $i_c+1$ (note that $\nu$ ends at $i_r-1$) is used by $\Au_{y^{\abs}_b}$ to remember that the clock $z_{\prec u}$ cannot be reset along $\nu$. Intuitively, live upper-bound obligations are propagated from the caller \MAP\ to the called \MAP. Note that fresh upper-bound obligations $(\first,\prec$$ u)$ always refer to predictions done along the current \MAP\ and, differently from the live upper-bound obligations, \mbox{they can be removed when the next $b$ occurs along the current \MAP.}

%To distinguish between
%finite and infinite \MAP, %and in particular to handle the infinite %\MAP\
%case, 
Extra technicalities are needed.
At each position $i$,
%the automaton
$\Au_{y^{\abs}_b}$ guesses whether $i$ is the last position of the current \MAP\ (i.e., the \MAP\ visiting $i$).
For this, it keeps track in its control state of the guessed type (call, return, or internal symbol) of the next input symbol. In particular, when $i$ is a call,
$\Au_{y^{\abs}_b}$ guesses whether %the call $i$
it has a matching return. If not, $\Au_{y^{\abs}_b}$ pushes onto the stack a special symbol, say $\bad$, and the guess is correct iff
the symbol is never popped from the stack. Conversely, $\Au_{y^{\abs}_b}$ exploits %a suitable encoding of 
a special proposition $p_{\infty}$ whose Boolean value is carried in the control state:
$p_{\infty}$ \emph{does not hold} at a position $j$ of the input iff the \MAP\ visiting $j$ has a caller whose matching return exists. %starts at a position $j>0$ such that $j-1$ is a matched call (i.e., position $i$ has a caller whose
%matching return exists).
Note that $p_{\infty}$ holds at infinitely many positions. The transition function of $\Au_{y^{\abs}_b}$ ensures that the Boolean value of
$p_{\infty}$ is propagated consistently with the guesses. Doing so,  the guesses about the matched calls are correct iff $p_{\infty}$ is asserted infinitely often along a run.
A B\"{u}chi component of $\Au_{y^{\abs}_b}$ ensures this last requirement.
Finally, we have to ensure that the lower-bound obligations and fresh upper-bound obligations at the current position are eventually checked, i.e., the current \MAP\ eventually visits a $b$-position.
For  finite \MAP, this can be %suitably
ensured by the transition function of $\Au_{y^{\abs}_b}$.
For infinite \MAP, we note that at most one
%infinite  \MAP\ $\nu$ exists along a word, and $\nu$ visits only positions where $p_{\infty}$ holds; moreover,   each position $i$ greatest than the initial position
%$i_0$ of $\nu$ is either a $\nu$-position, or a position where $p_{\infty}$ does not hold. Thus, a B\"{u}chi component of $\Au_{y^{\abs}_b}$ using   proposition $p_{\infty}$ ensures
%the $b$-liveness requirements along the unique infinite \MAP\ (if any). Full details of the construction of $\Au_{y^{\abs}_b}$ are in Appendix~\ref{APP:RemovalAbstractPredictor}.
%
%
%The transition function of $\Au_{y^{\abs}_b}$ guarantees the previous invariant along \MAP\ and additionally the following: (i) at a call position $i_c$, whenever $\Au_{y^{\abs}_b}$ guesses that $i_c$ has a matching return and $i_c+1$ is not a return, $\Au_{y^{\abs}_b}$ moves to a state where $\neg p_{\infty}$ is asserted, and (ii) on reading unmatched return positions (the symbol $\top$ is popped from the stack),
% $p_{\infty}$ is required to hold at the current state.
%Now we are able to handle the case of checking that a current infinite  \MAP\ eventually visits a $b$-position. We preliminary observe that
% at most one 
 infinite  \MAP\ $\nu$ exists along a word, and $\nu$ visits only positions where $p_{\infty}$ holds. Moreover,   each position $i$ greater than the initial position
$i_0$ of $\nu$ is either a $\nu$-position, or a position where $p_{\infty}$ does not hold. Thus, a B\"{u}chi component of $\Au_{y^{\abs}_b}$ using   proposition $p_{\infty}$ ensures
the $b$-liveness requirements along the unique infinite \MAP\ (if any). Full details of the construction of $\Au_{y^{\abs}_b}$ are in Appendix~\ref{APP:RemovalAbstractPredictor}.

By exploiting Theorems~\ref{theorem:closureProperties} and~\ref{theorem:FromECNAtoVPTA}, we establish the main result of the paper.

\begin{theorem} Emptiness,  universality, and language inclusion for B\"{u}chi \ECNA, and visibly model-checking against B\"{u}chi \ECNA\ are \EXPTIME-complete.
\end{theorem}
\begin{proof}
For the upper bounds, first observe (\cite{BouajjaniER94}) that the emptiness problem of generalized B\"{u}chi \VPTA\  is \EXPTIME-complete and solvable in time $O(n^{4} \cdot 2^{O(m\cdot \log K m)})$, where
$n$ is the number of states, $m$ is the number of clocks, and $K$ is the largest constant used in the clock constraints of the automaton.
%(hence, the time complexity is polynomial in the number of states).
Now, given two  B\"{u}chi \ECNA\ $\Au_1$ and $\Au_2$ over $\Sigma$, checking whether $\TLang(\Au_1)\subseteq \TLang(\Au_2)$ reduces to check emptiness of the language
$\TLang(\Au_1)\cap \overline{\TLang(\Au_2)}$. Similarly,  given a B\"{u}chi \VPTA\ $\mathcal{S}$ where all the states are accepting and a  B\"{u}chi \ECNA\ $\Au$ over the same pushdown alphabet $\Sigma$, model-checking
$\mathcal{S}$ against $\Au$ reduces to check emptiness of the language $\TLang(\mathcal{S})\cap \overline{\TLang(\Au)}$.
Since B\"{u}chi \VPTA\ are polynomial-time closed under intersection %(given two B\"{u}chi \VPTA, one can construct the `intersection' automaton in polynomial time)
and universality can be reduced in linear-time to
language inclusion, by the closure properties of B\"{u}chi \ECNA\ (Theorem~\ref{theorem:closureProperties}) and  Theorem~\ref{theorem:FromECNAtoVPTA}, membership in \EXPTIME\ for the considered problems directly follow.

For the matching lower-bounds, the proof of \EXPTIME-hardness for emptiness of B\"{u}chi \VPTA\ can be easily adapted to the class of B\"{u}chi \ECNA. For the other %considered
problems, the result directly follows
from \EXPTIME-hardness of the corresponding problems for B\"{u}chi \VPA~\cite{AlurMadhu04,AlurMadhu09} which are subsumed by B\"{u}chi \ECNA.\qed
\end{proof}

\paragraph{\bf Conclusions.} In this paper we have introduced and studied \ECNA, a robust subclass of \VPTA\ allowing to express meaningful non-regular timed properties of recursive systems. The closure under Boolean operations, and the decidability of  languages inclusion and visibly model-checking makes \ECNA\ amenable to specification and verification purposes. As future work, we plan to investigate suitable extensions of the Event Clock Temporal Logic introduced for \ECA\ so that a logical counterpart for \ECNA\ can be similarly recovered.

%\clearpage
\bibliography{bib2}
\bibliographystyle{splncs03}

\newpage

\appendix

\newcounter{aux}

\begin{center}
\begin{LARGE}
  \noindent\textbf{Appendix}
\end{LARGE}
\end{center}

\section{Proof of Theorem~\ref{theorem:closureProperties}}\label{APP:closureProperties}

In this section, we provide a proof of the following result.

\setcounter{aux}{\value{theorem}}
\setcounter{theorem}{\value{theo-closureProperties}}

\begin{theorem}[Closure properties] The class of $\omega$-timed languages accepted by B\"{u}chi
\ECNA\ is closed under union, intersection, and complementation. In particular, given two  B\"{u}chi
\ECNA\ $\Au=\tpl{\Sigma, Q,Q_{0},C,\Gamma\cup\{\top\},\Delta,F}$  and $\Au'=\tpl{\Sigma, Q',Q'_{0},C',\Gamma'\cup\{\top\},\Delta',F'}$ over $\Sigma$, one can construct
\begin{compactitem}
  \item a B\"{u}chi \ECNA\ accepting $\TLang(\Au)\cup \TLang(\Au')$ with $|Q|+|Q'|$ states, $|\Gamma|+|\Gamma'|+1$ stacks symbols, and greatest constant $\max(K_{\Au},K_{\Au'})$;
    \item a B\"{u}chi \ECNA\ accepting $\TLang(\Au)\cap \TLang(\Au')$ with $2|Q||Q'|$ states, $|\Gamma|
    |\Gamma'|$ stacks symbols,  and greatest constant $\max(K_{\Au},K_{\Au'})$;
      \item a B\"{u}chi \ECNA\ accepting the complement of $\TLang(\Au)$ with $2^{O(n^{2})}$ states,  $O(2^{O(n^{2})}\cdot |\Scall|\cdot |\Const|^{O(|\Sigma|)})$ stack symbols, and greatest constant
      $K_\Au$, where $n=|Q|$ and $\Const$ is the set of constants used in the clock constraints of $\Au$.
\end{compactitem}
\end{theorem}
\setcounter{theorem}{\value{aux}}

Let $\Au=\tpl{\Sigma, Q,Q_{0},C,\Gamma\cup\{\top\},\Delta,F}$  and $\Au'=\tpl{\Sigma, Q',Q'_{0},C',\Gamma'\cup\{\top\},\Delta',F'}$ be
two B\"{u}chi
\ECNA\ over $\Sigma$. Closure under union and intersection easily follows from the  language closure properties of B\"{u}chi \ECA\ and B\"{u}chi visibly pushdown automata (\VPA). In particular, the B\"{u}chi \ECNA\ accepting $\TLang(\Au)\cup \TLang(\Au')$ is obtained   by taking the
union of the states, stack symbols,  and transitions of $\Au$ and $\Au'$ (assuming they are disjoint) and taking the new
set of initial states (resp., final states) to be the union of the initial states (resp., final states) of $\Au$ and $\Au'$.
The B\"{u}chi \ECNA\ accepting $\TLang(\Au)\cap \TLang(\Au')$ has set of states $Q\times Q'\times \{0,1\}$, set of initial states
$Q_0\times Q'_0\times \{0\}$, stack alphabet $(\Gamma\times \Gamma')\cup\{\top\}$, and set of accepting states $Q\times F'\times \{1\}$.
 When reading a call, if $\Au$ pushes
$\gamma$ and $\Au'$ pushes
$\gamma'$, then the product automaton pushes $(\gamma,\gamma')$.
We exploit the fact
that $\Au$ and $\Au'$, being $\ECNA$ over the same pushdown alphabet, synchronize on the push and pop operations on the stack. The additional flag in the set of states is used to ensure  that the final states of both
automata are visited infinitely often.

It remains to prove the closure under language complementation. For this, we adopt the approach exploited in~\cite{TangO09} for the subclass of B\"{u}chi $\ECVPA$. In particular, we define an homomorphism from B\"{u}chi \ECNA\ to B\"{u}chi \VPA\ and vice versa. Note that a B\"{u}chi  \VPA\ is defined as a B\"{u}chi  \ECNA\ but we omit the set of event clocks, and the set of clock constraints from the transition function.
The notion of (accepting) run of a B\"{u}chi  \VPA\
 over an infinite word on $\Sigma$ is similar to the notion of (accepting) run of an \ECNA\ over an infinite timed word on $\Sigma$, but we omit the requirements about the clock constraints.

Fix a B\"{u}chi \ECNA\ $\Au=\tpl{\Sigma, Q,Q_{0},C,\Gamma\cup\{\top\},\Delta,F}$ and let $\Const =
\{c_0,\ldots,c_k\}$ be the set of  constants used in the clock constraints of $\Au$ ordered for increasing values, i.e. such that
$0\leq c_0<c_1\ldots <c_k$. We consider the following set $\Intv$ of intervals over $\RealP\cup \{\NULL\}$:
\[
\Intv := \{[\NULL,\NULL],[0,0],(0,c_0)\}\cup \bigcup_{i=0}^{i=k-1}\{[c_i,c_i],(c_i,c_{i+1})\}\cup \{[c_k,c_k],(c_k,\infty)\}
\]
A \emph{region $\reg$ of $\Au$} is a mapping $\reg: C \mapsto \Intv$ assigning to each event clock in $C\subseteq C_\Sigma$ an interval in $\Intv$.
The mapping $\reg$ induces the clock constraint $\bigwedge_{z\in C} z\in \reg(z)$. We denote by $[\reg]$ the set of valuations over $C$
satisfying the clock constraint associated with $\reg$, and by  $\Reg$ the set of regions of $\Au$. For a clock constraint $\theta$ over $C$, let
$[\theta]$ be the set of valuations over $C$ satisfying $\theta$.

\begin{remark}\label{remark:regions} By construction, the following holds.
\begin{compactitem}
  \item The set $\Reg$ of regions represents a partition of the set of clock valuations over $C$, i.e.: (i) for all valuations $\val$ over $C$, there is a region $\reg\in\Reg$
  such that $\val\in \Reg$, and (ii)
for all regions $\reg,\reg'\in\Reg$, $\reg\neq \reg' \Rightarrow [\reg]\cap [\reg']=\emptyset$.
\item for each clock constraint $\theta$ of $\Au$ and region $\reg\in\Reg$, either $[\reg]\subseteq [\theta]$ or
$[\reg]\cap [\theta]=\emptyset$.
\end{compactitem}
\end{remark}

We associate with $\Sigma= \Scall\cup\Sret\cup\Sint$ and the set of regions $\Reg$ a pushdown alphabet $\Lambda = \Sigma\times \Reg$, called \emph{interval pushdown alphabet},
whose set of calls is $\Scall\times \Reg$, whose set of returns is $\Sret\times \Reg$, and whose set of internal actions is $\Sint\times \Reg$. Elements
of $\Lambda$ are pairs of the form $(a,\reg)$, where $a\in\Sigma$ and $\reg$ is a region of $\Au$  which is meant to represent the associated constraint $\bigwedge_{z\in C} z\in \reg(z)$.
An infinite word $\lambda = (a_0,\reg_0),(a_1,\reg_1),\ldots$ over $\Lambda$ induces in a natural way a set of infinite timed words over $\Lambda$, denoted $tw(\lambda)$, defined as follows:
$w=(\sigma,\tau)\in\tw(\lambda)$ iff $\sigma= a_0 a_1\ldots$ and for all $i\geq 0$, $\val_{i}^{w}\in [\reg_i]$. We extend the mapping $tw$ to $\omega$-languages $\Lang$ over $\Lambda$ in the obvious way: $\tw(\Lang):= \bigcup_{\lambda\in\Lang}\tw(\lambda)$. By means of the mapping $\tw$, infinite words over $\Lambda$ define a partition of the set of infinite timed words
over $\Sigma$.

\begin{lemma}\label{lemma:symbolicTimedWords} The following holds.
\begin{compactenum}
  \item For each infinite timed word $w= (\sigma,\tau) $ over $\Sigma$, there is an infinite word $\lambda$ over $\Lambda$
     of the form $(\sigma_0,\reg_0)(\sigma_1,\reg_1)$ such that $w\in \tw(\lambda)$.
   \item For all infinite words $\lambda$ and $\lambda'$ over $\Lambda$, $\lambda \neq \lambda' \Rightarrow \tw(\lambda)\cap \tw(\lambda')=\emptyset$.
\end{compactenum}
\end{lemma}
\begin{proof} For Property~1, let $w= (\sigma,\tau) $ be an infinite timed word over $\Lambda$. By Remark~\ref{remark:regions}, for all $i\geq 0$, there is a region $\reg_i\in \Reg$
such that $\val_i^{w}\in [\reg_i]$. Let $\lambda = (\sigma_0,\reg_0)(\sigma_1,\reg_1)\ldots$. We have that
$w\in \tw(\lambda)$, and the result follows.

For Property~2, let $\lambda$ and $\lambda'$ be two distinct infinite words over $\lambda$. Let us assume that
$\tw(\lambda)\cap \tw(\lambda')\neq\emptyset$ and derive a contradiction. Hence, by construction,
$\lambda =(a_0,\reg_0)(a_1,\reg_1)\ldots$, $\lambda'=(a_0,\reg'_0)(a_1,\reg'_1)\ldots$, and there is an infinite timed word $w$ over $\Lambda$
of the form $(a_0,\tau_0)(a_1,\tau_1)$ such that $\val_i^{w}\in [\reg_i]\cap [\reg'_i]$ for all $i\geq 0$.
Since $\lambda\neq \lambda'$, there exists $n\geq 0$ such that $\reg_n\neq \reg'_n$. By Remark~\ref{remark:regions}, $[\reg_n]\cap [\reg'_n]=\emptyset$ which is a contradiction since
$\val_n^{w}\in [\reg_n]\cap [\reg'_n]$, and the result follows. \qed
\end{proof}

The following two propositions, establish an untimed homomorphism from B\"{u}chi \ECNA\ to B\"{u}chi \VPA, and a timed homomorphism from
  B\"{u}chi \VPA\ to B\"{u}chi \ECNA, respectively.

\begin{proposition}[Untimed homomorphism]\label{prop:UntimedHom} Let $\Au=\tpl{\Sigma, Q,Q_{0},C,\Gamma\cup\{\top\},\Delta,F}$ be a B\"{u}chi \ECNA, and $\Lambda$ be the interval pushdown alphabet induced by $\Au$. Then, one can construct  a B\"{u}chi \VPA\ $\Untimed(\Au)$ over $\Lambda$ of the form  $\tpl{\Lambda, Q,Q_{0},\Gamma\cup\{\top\},\Delta',F}$
such that $\tw(\Lang(\Untimed(\Au)))=\TLang(\Au)$.
\end{proposition}
\begin{proof} The transition function $\Delta'$ of $\Untimed(\Au)$ is defined as follows:
\begin{itemize}
\item Push: If $(q,c,\theta,q',\gamma)$ is a push transition in $\Delta$, then for each region
$\reg$ of $\Au$ such that $[\reg]\subseteq [\theta]$, $(q,(c,\reg),q',\gamma)\in\Delta'$.
\item Pop: If $(q,r,\theta,\gamma,q')$ is a pop transition in $\Delta$, then for each region
$\reg$ of $\Au$ such that $[\reg]\subseteq [\theta]$, $(q,(r,\reg),\gamma,q')\in\Delta'$.%
\item Internal:  If $(q,a,\theta,q')$ is an internal transition in $\Delta$, then for each region
$\reg$ of $\Au$ such that $[\reg]\subseteq [\theta]$, $(q,(a,\reg),q')\in\Delta'$.
\end{itemize}
By Remark~\ref{remark:regions} and Lemma~\ref{lemma:symbolicTimedWords}(1), we easily derive the correctness of the construction.\qed
\end{proof}

\begin{proposition}[Timed homomorphism]\label{prop:TimedHom} Let $\Au=\tpl{\Lambda, Q,Q_{0},\Gamma\cup\{\top\},\Delta,F}$ be a B\"{u}chi \VPA\ over an interval pushdown alphabet associated with $\Sigma$ and a set $C\subseteq C_\Sigma$ of event clocks. Then, one can construct  a B\"{u}chi \ECNA\ $\Timed(\Au)$ over $\Sigma$ of the form  $\tpl{\Sigma, Q,Q_{0},C,\Gamma\cup\{\top\},\Delta',F}$
such that $\TLang(\Timed(\Au))=\tw(\Lang(\Au))$.
\end{proposition}
\begin{proof} The transition function $\Delta'$ of $\Timed(\Au)$ is defined as follows:
\begin{itemize}
\item Push: If $(q,(c,\reg),q',\gamma)$ is a push transition in $\Delta$, then $(q,c,\theta,q',\gamma)\in\Delta'$
where $\theta := \bigwedge_{z\in C}z\in \reg(z)$.
\item Pop: If $(q,(r,\reg),\gamma,q')$ is a pop transition in $\Delta$, then $(q,r,\theta,\gamma,q')\in\Delta'$
where $\theta := \bigwedge_{z\in C}z\in \reg(z)$.%
\item Internal:  If $(q,(r,\reg),q')$ is an internal transition in $\Delta$, then $(q,r,\theta,q')\in\Delta'$
where $\theta := \bigwedge_{z\in C}z\in \reg(z)$.
\end{itemize}
By Remark~\ref{remark:regions} and Lemma~\ref{lemma:symbolicTimedWords}(1), we easily derive the correctness of the construction.\qed
\end{proof}

By Lemma~\ref{lemma:symbolicTimedWords}, Propositions~\ref{prop:UntimedHom} and~\ref{prop:TimedHom}, and the known  closure properties of B\"{u}chi Visibly Pushdown Automata (\VPA)~\cite{AlurMadhu04,AlurMadhu09}, we have
the following result.

\begin{theorem}[Closure under complementation]\label{theor:Complementation} Given a B\"{u}chi \ECNA\ $\Au$ over $\Sigma$ with $n$ states and set of constants $\Const$, one can construct in singly exponential time
a B\"{u}chi \ECNA\ $\Au$ over $\Sigma$ accepting the complement of $\TLang(\Au)$ having $2^{O(n^{2})}$ states and $O(2^{O(n^{2})}\cdot |\Scall|\cdot |\Const|^{O(|\Sigma|)})$ stack symbols.
\end{theorem}
\begin{proof}
Let $\Au$ be a B\"{u}chi \ECNA\ over $\Sigma$ with $n$ states and set of integer constants $\Const$, $\Lambda$ be the interval pushdown alphabet induced by $\Au$, and $\Lambda_c$ be the set of calls in $\Lambda$.
By Proposition~\ref{prop:UntimedHom}, we can construct a B\"{u}chi \VPA\ $\Untimed(\Au)$ over $\Lambda$ with $n$ states such that $\tw(\Lang(\Untimed(\Au)))= \TLang(\Au)$. By~\cite{AlurMadhu04,AlurMadhu09}, starting from
the B\"{u}chi \VPA\ $\Untimed(\Au)$, one can construct in singly exponential time
a B\"{u}chi \VPA\ $\overline{\Untimed(\Au)}$ over $\Lambda$ accepting $\Lambda^{\omega}\setminus \Lang(\Untimed(\Au))$ with $2^{O(n^{2})}$ states and $O(2^{O(n^{2})}\cdot |\Lambda_c|)$ stack symbols. Applying Proposition~\ref{prop:TimedHom} to the B\"{u}chi \VPA\ $\overline{\Untimed(\Au)}$, one can construct in linear time a B\"{u}chi \ECNA\ $\overline{\Au}$ over $\Sigma$ with $2^{O(n^{2})}$ states and $O(2^{O(n^{2})}\cdot |\Lambda_c|)$ stack symbols such that
$\TLang(\overline{\Au}) = \tw(\Lambda^{\omega}\setminus \Lang(\Untimed(\Au)))$. Since $\TLang(\Au)=\tw(\Lang(\Untimed(\Au)))$, by Lemma~\ref{lemma:symbolicTimedWords}, $\overline{\Au}$ accepts all and only the infinite timed words over $\Sigma$ which are not in $\TLang(\Au)$. Thus, since $|\Lambda_c|=O(|\Scall|\cdot |\Const|^{O(|\Sigma|)}) $, the result follows.\qed

\end{proof}

\section{Inexpressiveness results: proof of Proposition~\ref{prop:expressivenessPredRec}}\label{APP:expressivenessPredRec}

Let $\Sigma= \Scall\cup \Sret\cup \Sint$ be the pushdown alphabet
 with $\Scall =\{c\}$, $\Sret =\{r\}$, and $\Sint=\{a,b,\i\}$.
Let us consider the timed languages $\TLangRec$, $\TLangPred$, and $\TLangCaller$  over $\Sigma$ in Example~\ref{example:ECNA}. We show the following result.

\setcounter{aux}{\value{proposition}}
\setcounter{proposition}{\value{prop-expressivenessPredRec}}

\begin{proposition} The language $\TLangRec$ is \emph{not} definable by  B\"{u}chi \ECNA\ which do \emph{not} use abstract recorder clocks,  $\TLangPred$  is \emph{not} definable by B\"{u}chi \ECNA\ which do \emph{not} use abstract predictor clocks, and $\TLangCaller$ is \emph{not} definable by   B\"{u}chi \ECNA\ which do \emph{not} use caller clocks. %  an abstract-predicting (resp., abstract-recording) B\"{u}chi \ECNA.
Moreover, the language $\TLangRec\cup \TLangPred\cup \TLangCaller$ is \emph{not} definable by  B\"{u}chi \ARCNA,
   B\"{u}chi \APCNA\ and   B\"{u}chi \CECNA.
\end{proposition}

\setcounter{proposition}{\value{aux}}

\begin{proof} First, let us consider the timed language $\TLangRec$.  Recalling Example~\ref{example:ECNA}, $\TLangRec$   consists of all the timed words of the form $(\sigma,\tau)\cdot (\i^{\omega},\tau')$ such that $\sigma$ is a well-matched word of the form $a \cdot c^{+} \cdot a^{+} \cdot r^{+} \cdot b^{+}$ and the time difference in $(\sigma,\tau)$  between the first and last symbols is $1$. Let $v_1$ and $v_2$ be the finite timed words over $\Sigma$ of length $6$ defined as follows.
\begin{itemize}
  \item $v_1 = (a,0)\cdot (c,0.1) \cdot (a,0.1) \cdot (r,0.1) \cdot (b,0.1) \cdot (b,0.9)$.
  \item $v_2 = (a,0)\cdot (c,0.1) \cdot (a,0.1) \cdot (r,0.1) \cdot (b,0.1) \cdot (b,1)$.
\end{itemize}
For each $H\geq 1$, let $w_1^{H}= v_1\cdot (\i,H+2) \cdot (\i,H+3) \ldots$ and $w_2^{H}= v_2\cdot (\i,H+2) \cdot (\i,H+3) \ldots$. Let us denote by
$\val^{1,H}$ and $\val^{2,H}$ the event-clock  valuations over $C_\Sigma$
   associated with $w_1^{H}$ and $w_2^{H}$, respectively. By construction, the following easily follows for all positions $i\geq 0$ and   event-clocks $z\in C_\Sigma$ such that $z$ is \emph{not} an abstract recorder clock:
\begin{compactitem}
  \item either (i)  $\val^{1,H}_i(z)=\val^{2,H}_i(z)$, or (ii) $0<\val^{1,H}_i(z)<1$ and $0<\val^{2,H}_i(z)<1$, or (iii) $\val^{1,H}_i(z)>H$ and $\val^{2,H}_i(z)>H$.
\end{compactitem}
   Hence, clock constraints which do not use abstract recorder clocks and whose maximum constant is at most $H$ cannot distinguish the valuations $\val^{1,H}$ and $\val^{2,H}$. It follows that for each
    \ECNA\ $\Au$ over $\Sigma$ which does not use abstract recorder clocks and  has maximum constant $H$, $w_1^{H}\in \TLang(\Au)$ iff $w_2^{H}\in \TLang(\Au)$. On the other hand, by definition of the language  $\TLangRec$, for each $H\geq 1$,
   $w_2^{H}\in \TLangRec$ and $w_1^{H}\notin \TLangRec$. Hence, $\TLangRec$  is not definable by B\"{u}chi \ECNA\ which do not use abstract recorder clocks.\vspace{0.2cm}

Now, let us consider the timed language $\TLangPred$. Recall that $\TLangPred$  consists of all the timed words of the form $(\sigma,\tau)\cdot (\i^{\omega},\tau')$ such that $\sigma$ is a well-matched word of the form $a^{+}  \cdot  c^{+} \cdot b^{+} \cdot r^{+} \cdot b$ and the time difference in $(\sigma,\tau)$  between the two extreme symbols is $1$.
 Let $u_1$ and $u_2$ be the finite timed words over $\Sigma$ of length $6$ defined as follows.
\begin{itemize}
  \item $u_1 = (a,0)\cdot (a,0.1) \cdot (c,0.1) \cdot (b,0.1) \cdot (r,0.1) \cdot   (b,0.9)$.
  \item $u_2 = (a,0) \cdot (a,0.1) \cdot (c,0.1) \cdot (b,0.1) \cdot (r,0.1)   \cdot (b,1)$.
\end{itemize}
For each $H\geq 1$, let $r_1^{H}= u_1\cdot (\i,H+2) \cdot (\i,H+3) \ldots$ and $r_2^{H}= u_2\cdot (\i,H+2) \cdot (\i,H+3) \ldots$. By reasoning as for the the case of the language $\TLangRec$,  it easily follows that for each
 \ECNA\ $\Au$ over $\Sigma$ which does not use abstract predictor clocks and  has as maximum constant $H$, $r_1^{H}\in \TLang(\Au)$ iff $r_2^{H}\in \TLang(\Au)$. On the other hand, by definition of the language  $\TLangPred$, for each $H\geq 1$,
   $r_2^{H}\in \TLangPred$ and $r_1^{H}\notin \TLangPred$. Hence, $\TLangPred$  is not definable by   B\"{u}chi \ECNA\ which do not use abstract predictor clocks.

    The proof for the timed language $\TLangCaller$ is similar. Finally, we observe that by the above considerations, it follows that  $\TLangRec\cup \TLangPred\cup \TLangCaller$ is not definable neither by an  abstract-predicting  B\"{u}chi \ECNA\ nor
by an abstract-recording B\"{u}chi \ECNA\ nor by a caller B\"{u}chi \ECNA.\qed
\end{proof}

\section{Removal of abstract predictor clocks in nested \VPTA}\label{APP:RemovalAbstractPredictor}

In this section, we provide the details of the construction of the generalized B\"{u}chi nested  \VPTA\ $\Au_{y^{\abs}_b }$ described in Section~\ref{sec:DecisionProcedures} starting
from a generalized B\"{u}chi nested \VPTA\ $\Au=\tpl{\Sigma, Q,Q_{0},C\cup C_{st},\Gamma\cup\{\top\},\Delta,\mathcal{F}}$ such that the abstract predictor clock $x^{\abs}_b$ is in $C$.
For this, we need additional notation.

An \emph{obligation set} $O$ (for the fixed abstract predictor clock $y^{\abs}_b$ and the fixed generalized B\"{u}chi nested \VPTA\ $\Au$) is a set consisting of lower-bound
obligations $\succ $$\ell$ and upper-bound obligations  ($f$,$\prec$$u$), where $f\in \{\live,\first\}$, such that
$y_b^{\abs} \succ \ell$ and $y_b^{\abs} \prec u$ are associated to interval constraints of $\Au$, and $(f,\prec$$u), (f',\prec$$u)\in O$ implies $f=f'$.
For an obligation set $O$, $\live(O)$ is the obligation set consisting of the live upper-bound obligations of $O$.

Let us consider the \CARET\ formula \cite{AlurEM04} $\Eventually^{\abs} b$: $\Eventually^{\abs} b$ holds at position $i\geq 0$ if the \MAP\ visiting $i$ also visits a position $j\geq i$ where $b$ holds.  A \emph{check set} $H$ is a subset of $\{\call,\ret,\intA,b,\Eventually^{\abs} b,p_{\infty}\}$ such that
$H\cap \{\call,\ret,\intA\}$ is a singleton. Intuitively, a check set is exploited by $\Au_{y^{\abs}_b}$ for keeping track of: (i) the guessed type (call, return, or internal symbol) of the next
input symbol, (ii) whether the next input symbol is $b$, (iii) whether $p_{\infty}$ holds at the current position, and (iv) whether
 $\Eventually^{\abs} b$ holds at the current position.

Let $C_\new$ be the set of standard clocks consisting of the  fresh standard clocks $z_{\succ \ell}$ (resp., $z_{\prec u}$) for each lower-bound constraint  $y_b^{\abs}\succ \ell$  (resp., upper-bound constraint $y_b^{\abs} \prec u$) of $\Au$ involving $y_b^{\abs}$.
For an input symbol $a\in \Sigma$ and an obligation set $O$, we denote by $\con(O,a)$ the constraint over the new set $C_\new$ of standard clocks defined as:
$\con(O,a)=\true$ if either $O=\emptyset$ or $b\neq a$; otherwise, $\con(O,a)$ is obtained from $O$ by adding for each obligation $\succ \ell$ (resp., $(f,\prec$$ u)$) in $O$, the conjunct $z_{\succ \ell}\succ \ell$
(resp., $z_{\prec u}\prec u$).
The nested \VPTA\ $\Au_{y^{\abs}_b}$ is given by
\[
\Au_{y^{\abs}_b}= \tpl{\Sigma, Q',Q'_{0},C\setminus\{y^{\abs}_b\}\cup C_{st} \cup C_{\new},(\Gamma\times Q')\cup\{\bad,\top\},\Delta',\mathcal{F}'}
 \]
 The set $Q'$ of states
 consists of triples of the form $(q,O,H)$ such that $q$ is a state of $\Au$,
$O$ is an obligation set, and $H$ is a check set, while the set $Q'_0$ of initial states consists of states of the form $(q_0,\emptyset,H)$ such that $q_0\in Q_0$ (initially there are no obligations).

We now define the transition function $\Delta'$. For this, we first define a predicate $\Abs$ over tuples of the form $\tpl{(O,H),a, y_b^{\abs}\in I,\Res,(O',H')}$
where $(O,H),(O',H')$ are pairs of obligation sets and check sets, $a\in\Sigma$, $y_b^{\abs}\in I$ is a constraint of $\Au$ involving $y_b^{\abs}$, and $\Res\subseteq C_\new$. Intuitively, $O$ (resp., $H$) represents the obligation set (resp., check set)  at the current position $i$ of the input, $a$ is the input symbol associated with position $i$, $y_b^{\abs}\in I$ is the prediction about $y_b^{\abs}$ done by $\Au$ at position $i$,
$\Res$ is the set of new standard clocks reset by $\Au_{y^{\abs}_b}$ on reading $a$,
 and
$O'$ (resp., $H'$) represents the obligation set (resp., check set) at the position $j$ following $i$ along the \MAP\ visiting  $i$ (if $i$ is a call, then $j$ is the matching-return of $i$). Formally,
$\Abs((O,H),a, y_b^{\abs}\in I,\Res,(O',H'))$ iff the following holds:
  \begin{enumerate}
  \item ($p_{\infty}\in H$ iff $p_{\infty}\in H'$), $a\in \Scall$ (resp., $a\in \Sret$, resp. $a\in \Sint$) implies $\call\in H$ (resp., $\ret\in H$, resp., $\intA\in H$).
  \item $\Eventually^{\abs} b\in H$ iff ($b= a$ or $\Eventually^{\abs} b\in H'$), and ($\Eventually^{\abs} b \in H'$ iff $I\neq\{\NULL\}$).
  \item If $I=\{\NULL\}$, then   $O'=\live(O)$, $\Res =\emptyset$, and $b\neq a$ implies $O=\live(O)$. Otherwise, let
   $y_b^{\abs}\in I \equiv y_b^{\abs}\succ \ell \wedge y_b^{\abs} \prec u$.  Let $O''$ be $O$ if $b\neq a$, and $O''=\live(O)$ otherwise. Then, $O' = O'' \cup \{ \succ$$ \ell\} \cup \{ (f, \prec$$ u)\}$, where $f= \live$ if $(\live, \prec$$ u)\in O''$, and $f=\first$ otherwise. Moreover, $\Res\subseteq \{z_{\succ \ell},z_{\prec u}\}$, $z_{\succ \ell}\in\Res$, and $z_{\prec u}\in \Res$ iff either
    $\prec$$ u$ does not appear in $O$, or $b= a$ and $(\first,\prec$$ u)\in O$.
   \end{enumerate}
Condition~1 requires that the Boolean value of proposition $p_{\infty}$ is invariant along the positions of a \MAP, and
the current check set is consistent with the type (call, return, or internal symbol) of the current input symbol.
Condition~2 provides the abstract-local propagation rules of formula $\Eventually^{\abs} b$.
Finally, Condition~3 provides the rules for updating the obligations on moving to the abstract next position along the current \MAP\ and for resetting new clocks on reading the current input symbol
$a$. Note that if $I=\{\NULL\}$ and $b\neq a$, then the current obligation set must contain only live upper-bound obligations. If, instead,
$y_b^{\abs}\in I$ is equivalent to  $y_b^{\abs}\succ \ell \wedge y_b^{\abs} \prec u$, then   the clock $z_{\succ \ell}$ is  reset, while the clock
$z_{\prec u}$ is reset iff either there is no obligation $(f,\prec$$u)$ in $O$, or $b= a$  and the obligation $(f,\prec$$u)$ is fresh, i.e., $f=\first$.

  The transition function $\Delta'$ of $\Au_{y^{\abs}_b}$  is then defined as follows. Recall that we can assume that each clock constraint of $\Au$ is of the form
$\theta \wedge y_b^{\abs}\in I$, where $\theta$ does not contain occurrences of $y_b^{\abs}$.

\paragraph{Push transitions:} for each push transition $q\, \der{a,\theta \wedge y_b^{\abs}\in I, \Res,\push(\gamma)}\,q'$ of $\Au$, we have the push  transitions
$(q,O,H) \, \der{a,\theta\wedge \con(O,a), \Res\cup\Res',\push(\gamma')}\,(q',O',H')$ such that $b=a$ iff $b\in H$, and
\begin{enumerate}
\item Case $\gamma'\neq \bad $. Then,   $\gamma'=(\gamma,O_{\ret},H_{\ret})$  and
  \begin{compactitem}
 \item $\Abs((O,H),a,y_b^{\abs}\in I,\Res',(O_{\ret},H_{\ret}))$. Moreover, if $\ret\in H'$ then  $H_{\ret} =H'$ and $O'=O_{\ret}$; otherwise, $p_{\infty}\notin H'$  and $O'$ consists of the live  obligations
  $(\live,\prec $$u)$ such that  $(f,\prec $$u)\in O_{\ret}$ for some $f\in \{\live,\first\}$.
 \end{compactitem}
   \item Case $\gamma'= \bad$: $\call\in H$, $I=\{\NULL\}$, ($\Eventually^{\abs} b\in H$ iff $b= a$), $p_{\infty}\in H$, $p_{\infty}\in H'$, $\ret\notin H'$, $O'=\emptyset$, $\Res'=\emptyset$, and $b\neq  a$  implies $O=\emptyset$.
\end{enumerate}
Note that if $b= a$, the obligations in the current state are checked by the constraint on $C_\new$ given by $\con(O,a)$ (recall that if $b\neq a$, then $\con(O,a)=\true$). The push transitions of point~1 consider the case  where $\Au_{y^{\abs}_b}$ guesses that the current call position $i_c$ has a matching return $i_r$. In this case, the set of obligations and the check state for the next abstract position $i_r$ along the current \MAP\ are pushed on the stack in order to be recovered at the matching-return $i_r$. Moreover, if $\Au_{y^{\abs}_b}$ guesses that the next position $i_c+1$ is not $i_r$ (i.e., $\ret\notin H'$), then all the upper-bound obligations in $O_{\ret}$ are propagated as  live obligations at the next position $i_c+1$ (note that the \MAP\ visiting $i_c+1$ starts at $i_c+1$, terminates at $i_r-1$, and does not satisfy proposition $p_{\infty}$).
The push transitions of point~2 consider instead the case  where $\Au_{y^{\abs}_b}$ guesses that the current call position $i_c$ has no matching return $i_r$, i.e., $i_c$ is the last position of the current \MAP. In this case, $\Au_{y^{\abs}_b}$ pushes the symbol $\bad$ on the stack and the transition relation is consistently updated.

\paragraph{Internal transitions:} for each internal transition $q\, \der{a,\theta \wedge y_b^{\abs}\in I, \Res}\,q'$ of $\Au$, we add   the internal transitions
$(q,O,H) \, \der{a,\theta\wedge \con(O,a), \Res\cup\Res'}\,(q',O',H')$, where  $b=a$ iff $b\in H$, and
\begin{enumerate}
  \item Case $\ret\in H'$: $\intA\in H$,  $I=\{\NULL\}$, $\Res' = \emptyset$,  ($\Eventually^{\abs}b\in H $ iff $b= a$), and $b\neq  a$ implies $O=\live(O)$.
  \item Case $\ret\notin  H'$:  $\Abs((O,H),a, y_b^{\abs}\in I,\Res',(O',H'))$.
\end{enumerate}
In the first case, $\Au_{y^{\abs}_b}$ guesses that
the current internal position $i$ is the last one of the current \MAP ($\ret\in H'$), while in the second case  the current \MAP\ visits the next non-return position $i+1$.
Note that if $b= a$, the obligations in the current state are checked by the constraint $\con(O,a)$.
\paragraph{Pop transitions:} for each pop transition $q\, \der{a,\theta \wedge y_b^{\abs}\in I, \Res,\pop(\gamma)}\,q'\in \Delta_r$, we have the pop  transitions
$(q,O,H) \, \der{a,\theta\wedge \con(O,a), \Res\cup\Res',\pop(\gamma')}\,(q',O',H')$, where $b=a$ iff $b\in H$, and
\begin{enumerate}
  \item Case $\gamma\neq \top$: $\ret\in H$ and $\gamma'=(\gamma,(O,H))$.
  If $\ret\notin H'$, then  $\Abs((O,H),a, y_b^{\abs}\in I,\Res',(O',H'))$; otherwise, $I = \{\NULL\}$, $\Res'=\emptyset$,
  ($\Eventually^{\abs} b\in H$ iff $b= a$), and $b\neq  a$ implies $O= \live(O)$.
\item Case $\gamma = \top$: $\ret\in H$, $O=\emptyset$,  $\gamma'=\top$,   $p_{\infty}\in H$, and $p_{\infty}\in H'$.
If $\ret\notin H'$, then  $\Abs((O,H),a, y_b^{\abs}\in I,\Res',(O',H'))$; otherwise, $I = \{\NULL\}$, $\Res'=\emptyset$, $O'=\emptyset$, and
  ($\Eventually^{\abs}b\in H$ iff $b= a$).
\end{enumerate}
If $\gamma\neq \top$, then the current return position has a matched-call. Thus, $\Au_{y^{\abs}_b}$ pops from the stack $\gamma$ together with an obligation set and a check set, and verifies that the last two sets
correspond to the ones associated with the current control state.
If $\gamma = \top$, then the current position is also the initial position
of the associated \MAP.

Finally, the generalized B\"{u}chi condition $\mathcal{F}'$ of $\Au_{y^{\abs}_b}$ is defined as follows. For each B\"{u}chi component $F$ of $\Au$,
 $\Au_{y^{\abs}_b}$  has the B\"{u}chi component consisting of the states $(q,O,H)$ such that $q\in F$. Moreover,
  $\Au_{y^{\abs}_b}$ has an additional component consisting of the states $(q,O,H)$ such that $p_{\infty}\in H$, and either $\Eventually^{\abs}b\notin H$ or $b\in H$.
 Such a component  ensures that  the guesses about the matched calls are correct ($p_{\infty}$ occurs infinitely often), and that  the liveness requirement
$b$ of $\Eventually^{\abs}b$ is fulfilled whenever  $\Eventually^{\abs}b$ is asserted at a position of an infinite \MAP. Recall that in an infinite word over $\Sigma$, there are at most one
infinite  \MAP\ $\nu$ and $\nu$ visits only positions where $p_{\infty}$ holds; moreover,   each position $i$ greatest than the initial position
$i_0$ of $\nu$ is either a $\nu$-position, or a position where $p_{\infty}$ does not hold. If an infinite word has no infinite \MAP, then $p_{\infty}$ holds at infinitely many positions as well.

\section{Removal of abstract recorder clocks in nested \VPTA}\label{APP:RemovalAbstractRecorder}

In this section, we establish the following result.

\begin{theorem}[Removal of abstract recorder clocks] \label{theorem:removeAbstractRecorder}
Given a generalized B\"{u}chi nested  \VPTA\ $\Au$ with set of event clocks $C$   and an abstract recorder clock $x_b^{\abs}\in C$, one can construct in singly exponential time a generalized B\"{u}chi nested  \VPTA\ $\Au_{x_b^{\abs}}$ with set of  event clocks $C\setminus\{x_b^{\abs}\}$  such that
$\TLang(\Au_{x_b^{\abs}})=\TLang(\Au)$ and $K_{\Au_{x_b^{\abs}}}=K_{\Au}$. Moreover, $\Au_{x_b^{\abs}}$ has $O(n\cdot 2^{O(p)})$ states and $m + O(p)$ clocks, where $n$ is the number of $\Au$-states, $m$ is the number
of standard $\Au$-clocks, and $p$ is the number of \emph{event-clock} atomic constraints on  $x_b^{\abs}$ used by $\Au$.
\end{theorem}

In the following, we illustrate the proof of Theorem~\ref{theorem:removeAbstractRecorder}.
Fix a generalized B\"{u}chi nested \VPTA\ $\Au=\tpl{\Sigma, Q,Q_{0},C\cup C_{st},\Gamma\cup\{\top\},\Delta,\mathcal{F}}$ such that $x^{\abs}_b \in C$.
We can assume that for each transition $\delta$ of $\Au$, there is exactly one atomic constraint $x^{\abs}_b\in I$ on $x^{\abs}_b$ used as conjunct in the clock constraint
of $\delta$.
We construct a generalized B\"{u}chi nested \VPTA\ $\Au_{x^{\abs}_b }$ equivalent to $\Au$ whose set of event clocks is  $C\setminus \{x^{\abs}_b\}$, and whose set of standard clocks is
$C_{st}\cup C_\new$, where $C_\new$ consists of the  fresh standard clocks $z_{\succ \ell}$ (resp., $z_{\prec u}$) for each lower-bound  constraint  $x_b^{\abs}\succ \ell$  (resp., upper-bound constraint $x_b^{\abs} \prec u$)
of  $\Au$ involving $x_b^{\abs}$.

We first explain the basic ideas of the translation. Note that a global recorder clock $x_b^{\Global}$ can be trivially converted in a standard clock by resetting it whenever $b$ occurs
along the input word. This approach is not correct for the abstract recorder clock $x_b^{\abs}$, since along a \MAP\ $\nu$, there may be consecutive positions $i_c$ and $i_r$
such that $i_c$ is a call with matching return $i_r$, and $b$ may occur along positions in $[i_c+1,i_r-1]$ which are associated with \MAP\ distinct from $\nu$.
Thus, as in the case of the abstract predictor  clock $y_b^{\abs}$, we replace $x_b^{\abs}$ with the  set $C_\new$ of fresh standard clocks defined above.
For a given infinite word $\sigma$ over $\Sigma$, a \MAP\ $\nu$ of $\sigma$ and a position $i$ of $\nu$, we denote by $\infix_b(\nu,i)$  the infix of $\nu$ defined as follows:
if there exists the smallest $b$-position $j >i$ visited by $\nu$, then $\infix_b(\nu,i)$ is the infix of $\nu$ between the next position of $i$ along $\nu$ and the position $j$; otherwise,
 $\infix_b(\nu,i)$ is the suffix of $\nu$ starting from the next position of $i$ along $\nu$ (note that in this case  $\infix_b(\nu,i)$ is empty if $i$ is the last position of $\nu$).
 The main idea of the construction is that when $b$ occurs at the current position $i$ of the input word, the simulating automaton $\Au_{x^{\abs}_b}$ guesses the set of lower-bound and upper-bound constraints on
 $x_b^{\abs}$ which will be used by $\Au$ along the portion $\infix_b(\nu,i)$ of the current \MAP.

  First, let us consider lower-bound constraints $x_b^{\abs} \succ \ell$. Assume that $b$ occurs at position $i$ of the input word for the first time and that $i$ is not the last position of the current \MAP\
  $\nu$ (hence, $\infix_b(\nu,i)$ is not empty). Then, $\Au_{x^{\abs}_b}$ guesses the set of lower-bound constraints  $x_b^{\abs} \succ \ell$  which will be used by $\Au$ along $\infix_b(\nu,i)$.
  For each  of such guessed constraints $x_b^{\abs} \succ \ell$, the associated new clock $z_{\succ \ell}$ is reset; moreover,
    if $i$ is not a call (resp., $i$ is a call), $\Au_{x^{\abs}_b}$ carries the obligation $\succ$$\ell$ in its control state (resp., pushes the obligation $\succ$$\ell$ onto the stack).
    On visiting the positions $j$ in $\infix_b(\nu,i)$, $\Au_{x^{\abs}_b}$ checks that the guess is correct by verifying that for the current lower-boud constraint
    $x_b^{\abs} \succ \ell'$ used by $\Au$, $\succ$$\ell'$ is in the current set of obligations, and constraint $z_{\succ \ell'}\succ\ell'$ holds. Moreoever, at position $j$,
   $\Au_{x^{\abs}_b}$ guesses whether the constraint $x_b^{\abs} \succ \ell'$ will  be again used along   $\infix_b(\nu,i)$, or not. In the first case, the obligation $\succ$$\ell'$ is kept, otherwise, it is discarded.
The crucial observation is that:
\begin{compactitem}
  \item If a call  $i_c\geq i$ occurs along $\nu$ before the last position (if any) of $\infix_b(\nu,i)$, we know that the next position of $i_c$  along $\nu$ is the matching return $i_r$ of $i_c$, $i_r$ is visited by $\infix_b(\nu,i)$,  and all the \MAP\ visiting positions
  $h\in [i_c+1,i_r-1]$ are finite and terminate at positions $k<i_r$. Thus, the fulfillment of a lower-bound constraint  $x_b^{\abs} \succ \ell$ asserted at a position of such \MAP\ always implies the fulfillment of the same  constraint when asserted at a position $j\geq i_r$ of $\infix_b(\nu,i)$. Thus, at the time of a guess (i.e., when a $b$ occurs) along a \MAP\ visiting positions in $[i_c+1,i_r-1]$, the clocks $z_{\succ \ell}$
  associated with the guessed lower-bound constraints $x_b^{\abs} \succ \ell$
 can be safely reset.
\end{compactitem}\vspace{0.2cm}

\noindent   At each position $i$, $\Au_{x^{\abs}_b}$ keeps track in its control state
of the lower-bound obligations for the part $\infix_b(\nu,i)$ of the current \MAP\ $\nu$. Whenever  a call $i_c$ occurs, the guessed lower-bound obligations for the matching return $i_r$ of $i_c$ are pushed on the stack in order to be recovered at position $i_r$. Moreover, if $i_c+1$ is not a return (i.e., $i_r\neq i_c+1$), then $\Au_{x^{\abs}_b}$ moves to a control state where the set of lower-bound obligations is empty (consistently with the fact that $i_c+1$ is the initial position of the \MAP\ visiting $i_c+1$).

The case for  upper-bound constraints $x_b^{\abs}  \prec u$ is symmetric. Whenever   $b$ occurs at a position $i$ of the input word which is not the last position of the current \MAP\ $\nu$ and
$\Au_{x^{\abs}_b}$ guesses that the constraint $x_b^{\abs}  \prec u$ will be used by $\Au$ along the infix $\infix_b(\nu,i)$, then, assuming that no obligation is currently associated to the constraint
$x_b^{\abs}  \prec u$,  $\Au_{x^{\abs}_b}$ resets the standard clock $z_{\prec u}$ and carries the fresh obligation ($\first$,$\prec$$u$) in its control state (resp., pushes the obligation ($\first$,$\prec$$u$) onto the stack) if $i$ is not a call (resp., $i$ is a call). When  at a position $j$ of the infix $\infix_b(\nu,i)$, $\Au$ uses the constraint $x_b^{\abs} \prec u$, $\Au_{x^{\abs}_b}$ checks that $(\first,\prec$$u)$ is in the current set of obligations, and that the constraint $z_{\prec u}\prec u$ holds. The obligation ($\first$,$\prec$$u$) is removed or confirmed, depending on whether $\Au_{x^{\abs}_b}$ guesses that
$x_b^{\abs} \prec u$ will be again used by $\Au$ along $\infix_b(\nu,i)$ or not.  Assume now that a call position $i_c\geq i$ occurs along $\nu$ before   the last position (if any) of $\infix_b(\nu,i)$, and let $i_r$ be the matching return of $i_c$. The important observation is that:
\begin{compactitem}
  \item the fulfillment of an upper-bound constraint  $x_b^{\abs} \prec u$ asserted at a position   $j\geq i_r$ of $\infix_b(\nu,i)$ always implies the fulfillment  of the same constraint when asserted at a position $h$ of a \MAP\ $\nu'$ visiting positions in $[i_c+1,i_r-1]$ such that $h$ is preceded along $\nu'$ by a position where $b$ occurs. Thus, if the constraint $x_b^{\abs} \prec u$ is guessed to hold  at a position   $j\geq i_r$ of $\infix_b(\nu,i)$, for the guesses on the constraint $x_b^{\abs} \prec u$ done by $\Au_{x^{\abs}_b}$ along the positions in $[i_c+1,i_r-1]$, the clock $z_{\prec u}$ is not reset at the times of the  guesses (i.e., when $b$ occurs along the positions in $[i_c+1,i_r-1]$ ).
\end{compactitem}\vspace{0.2cm}

\noindent   Whenever a
call $i_c$ occurs, the updated set $O$ of  upper-bound and lower-bounds obligations   is pushed onto the stack  in order to be recovered at the matching return $i_r$ of $i_c$.
Moreover, if $i_c+1$ is not a return (i.e., $i_r\neq i_c+1$), then $\Au_{x^{\abs}_b}$ moves to a control state where, while the set of lower-bound obligations is empty, the set  of upper-bound  obligations is
obtained from $O$ by replacing each upper-bound obligation ($f$,$\prec$$u$), where $f\in \{\live,\first\}$ with the live obligation ($\live$,$\prec$$u$). A live obligation ($\live$,$\prec$$u$) asserted at the initial position $i_c+1$ of the \MAP\  $\nu$ visiting $i_c+1$ (note that $\nu$ terminates at position $i_r-1$) is used by $\Au_{x^{\abs}_b}$ to remind that the clock $z_{\prec u}$ cannot be reset along $\nu$ when $b$ occurs. Intuitively, live upper-bound obligations are propagated from the caller \MAP\ to the called \MAP. Note that fresh upper-bound obligations $(\first,\prec$$ u)$ always refer to guesses done along the current \MAP\, and differently from the live upper-bound obligations, may be removed, when along the current \MAP, they are checked.

There are other technical issues to be handled. As for the construction associated to the automaton $\Au_{y^{\abs}_b}$  for an abstract predictor clock $y^{\abs}_b$, the automaton
  $\Au_{x^{\abs}_b}$ uses the special proposition $p_{\infty}$, and keeps track in its control state  of the guessed type (call, return, or internal symbol) of the next input symbol in order to check whether the current input position
   is the last one of the current \MAP.
Moreover, we have to ensure that the lower-bound obligations $\succ$$\ell$ (resp., the fresh upper-bound obligations $(\first,\prec$$ u)$) at the current position  $i$ are eventually checked, i.e., for the current \MAP\ $\nu$, $\infix_b(\nu,i)$ eventually visits a position $j$  where
the constraint $x^{\abs}_b \succ\ell$ (resp., $x^{\abs}_b \prec u$) is used.
For this, $\Au_{x^{\abs}_b}$ keeps track in its control state  of the guessed interval constraint
$x^{\abs}_b \in I$ used by $\Au$ on reading the next input symbol, and whether the guessed next input symbol is $b$.
Moreover, for each lower-bound obligation $\succ$$\ell$ (resp., fresh upper-bound obligations $(\first,\prec$$ u)$),
  $\Au_{x^{\abs}_b}$ exploits a B\"{u}chi component ensuring that along an infinite \MAP\ $\nu$, \emph{either} there are infinitely many occurrences of $b$-positions,\emph{ or} there are infinitely many occurrences of positions
  where an interval constraint $x^{\abs}_b \in I$  consistent with $x^{\abs}_b \succ\ell$ (resp., $x^{\abs}_b \prec u$) is used, \emph{or} there are infinitely many positions in $\nu$ where the set of obligations
  does not contain $\succ$$\ell$ (resp., $(\first,\prec$$ u)$).

We now provide the formal definition of $\Au_{x^{\abs}_b}$. To this end, we need additional notation. An \emph{obligation set} $O$ (for the fixed recorder event $x^{\abs}_b$) is a set consisting of lower-bound
obligations $\succ $$\ell$ and upper-bound obligations  ($f$,$\prec$$u$), where $f\in \{\live,\first\}$, such that
$x_b^{\abs} \succ \ell$ and $x_b^{\abs} \prec u$ are associated to interval constraints $x_b^{\abs}\in I$  of $\Au$, and $(f,\prec$$u), (f',\prec$$u)\in O$ implies $f=f'$.
For an obligation set $O$, $\live(O)$  consists of the live upper-bound obligations of $O$. Given an obligation set $O$ and an interval constraint
$x_b^{\abs} \in I$ of $\Au$, we say that $x_b^{\abs} \in I$ is \emph{consistent with $O$} if one of the following holds:
\begin{itemize}
  \item $I=\{\NULL\}$ and $O=\live(O)$.
  \item $x_b^{\abs}\in I \equiv x_b^{\abs}\succ \ell \wedge x_b^{\abs} \prec u$,  $\succ $$ \ell\in O$ and  $(f,\prec $$u)\in O$ for some $f\in \{\first,\live\}$.
\end{itemize}
  Let $\Phi(x^{\abs}_b)$ be the set of interval constraints of the form
 $x_b^{\abs} \in I$ used by $\Au$.
A \emph{check set} $H$ is a subset of $\{\call,\ret,\intA,p_{\infty},b\}\cup \Phi(x^{\abs}_b)$ such that
$H\cap \{\call,\ret,\intA\}$ and $H\cap \Phi(x^{\abs}_b)$ are singletons. We say that $H$ and an obligation set $O$ are \emph{consistent} if
the unique interval constraint in $H$ is consistent with $O$.
 For an interval constraint
 $x_b^{\abs} \in I$ used by $\Au$, let $\con(I)$ be the constraint over $C_\new$ defined as follows:
 $\con(I)= \true$ if $I=\{\NULL\}$, and $\con(I)= z_{\succ\ell} \succ \ell \wedge z_{\prec u} \prec u$ if  $x_b^{\abs}\in I \equiv x_b^{\abs}\succ \ell \wedge x_b^{\abs} \prec u$.
The nested \VPTA\ $\Au_{x^{\abs}_b}$ is given by
\[
\Au_{x^{\abs}_b}= \tpl{\Sigma, Q',Q'_{0},C\setminus\{x^{\abs}_b\}\cup C_{st} \cup C_{\new},(\Gamma\times Q')\cup\{\bad,\top\},\Delta',\mathcal{F}'}
 \]
 where the set $Q'$ of states consists of triples of the form $(q,O,H)$ such that $q$ is a state of $\Au$,
$O$ is an obligation set, $H$ is a check set, and $H$ and $O$ are consistent. The set $Q'_0$ of initial states consists of states of the form $(q_0,\emptyset,H)$ such that $q_0\in Q_0$ (initially there are no obligations). Note that for an initial state $(q_0,\emptyset,H)$, $(x_b^{\abs} \in \{\NULL\})\in H$ ($H$ and the obligation set $\emptyset$ are consistent).

We now define the transition function $\Delta'$. To this end, we first define a predicate $\AbsP$ over tuples of the form $\tpl{(O,H),a, x_b^{\abs}\in I,\Res,(O',H')}$
where $(O,H),(O',H')$ are pairs of obligation sets and check sets, $a\in\Sigma$, $x_b^{\abs}\in I$ is a constraint of $\Au$, and $\Res\subseteq C_\new$. Intuitively, $O$ (resp., $H$) represents the obligation set (resp., check set)  at the current position $i$ of the input, $a$ is the input symbol associated with position $i$, $x_b^{\abs}\in I$ is the constraint on $x_b^{\abs}$ used   by $\Au$ at position $i$,
$\Res$ is the set of new standard clocks reset by $\Au_{x^{\abs}_b}$ on reading $a$,
 and
$O'$ (resp., $H'$) represents the obligation set (resp., check set) at the position $j$ following $i$ along the \MAP\ visiting  $i$ (if $i$ is a call, then $j$ is the matching-return of $i$). Formally,
$\AbsP((O,H),a, x_b^{\abs}\in I,\Res,(O',H'))$ is true iff the following holds:
  \begin{enumerate}
  \item ($p_{\infty}\in H$ iff $p_{\infty}\in H'$), $a\in \Scall$ (resp., $a\in \Sret$, resp. $a\in \Sint$) implies $\call\in H$ (resp., $\ret\in H$, resp., $\intA\in H$).
  \item  $(x_b^{\abs}\in I)\in H$, and $H$ and $O$ are consistent (resp., $H'$ and $O'$ are consistent).
  \item Case $b= a$: $b\in H$, $(x_b^{\abs}\in \{\NULL\})\notin H'$ and $\Res$ is a subset of $C_\new$ such that $z_{\prec u}\in \Res$ implies $(\live,\prec$$ u)\notin O$. Moreover, $O'=\live(O)\cup O''$, where $O''$
  is obtained from $\Res$ by adding for each clock $z_{\succ \ell}\in\Res$ (resp., $z_{\prec u}\in \Res$), the obligation $\succ$$\ell$ (resp., the fresh obligation $(\first,\prec $$u$)).
  \item Case $b\neq a$: $b\notin H$, $\Res =\emptyset$.   If $I=\{\NULL\}$, then   $O'=O=\live(O)$ and $(x_b^{\abs}\in \{\NULL\})\in H'$. Otherwise,  let $x_b^{\abs}\in I \equiv x_b^{\abs}\succ \ell \wedge x_b^{\abs} \prec u$. Then, $(x_b^{\abs}\notin \{\NULL\})\in H'$, and  $O'$ is any obligation set obtained
  from $O$ by \emph{optionally} removing the obligation $\succ $$\ell$ (by Condition~2, $\succ $$\ell\in O$), and/or by \emph{optionally} removing the obligation
  $(\first,\prec$$ u)$ if  $(\first,\prec$$ u)\in O$.
   \end{enumerate}
Condition~1 requires that the Boolean value of proposition $p_{\infty}$ is invariant along the positions of a \MAP, and
the current check set is consistent with the type (call, return, or internal symbol) of the current input symbol. Condition~2 requires that the current check set is consistent with the costraint $x_b^{\abs}\in I$ currently used by $\Au$.
 Conditions~3 and~4 provide  the rules for updating the obligations on moving to the abstract next position along the current \MAP\ and for resetting new clocks on reading the current input symbol
$a$. Note that if $I=\{\NULL\}$ and $b\neq  a$, then the current obligation set must contain only live upper-bound obligations, and $(x_b^{\abs}\in \{\NULL\})\in H'$.

Given a state $(q,O,H)$ of $\Au_{x^{\abs}_b}$, we say that $(q,O,H)$ is \emph{terminal} if the following holds: if $x_b^{\abs}\in I$ is the unique constraint associated
with the check set $H$ and $x_b^{\abs}\in I \equiv x_b^{\abs}\succ \ell \wedge x_b^{\abs} \prec u$, then $O\setminus \{\succ$$ \ell,(\first,\prec$$ u) \}=\live(O)$. Intuitively, terminal states are associated with
input positions $i$ such that $i$ is the last position of the related \MAP.

 The transition function $\Delta'$ of $\Au_{x^{\abs}_b}$  is then defined as follows. Recall that we can assume that each clock constraint of $\Au$ is of the form
$\theta \wedge x_b^{\abs}\in I$, where $\theta$ does not contain occurrences of $x_b^{\abs}$.

\paragraph{Push transitions:} for each push transition $q\, \der{a,\theta \wedge x_b^{\abs}\in I, \Res,\push(\gamma)}\,q'$ of $\Au$, we have the push  transitions
$(q,O,H) \, \der{a,\theta\wedge \con(I), \Res\cup\Res',\push(\gamma')}\,(q',O',H')$ such that $b\in H$ iff $a=b$, and
\begin{enumerate}
\item Case $\gamma'\neq \bad $. Then,   $\gamma'=(\gamma,O_{\ret},H_{\ret})$  and
  \begin{compactitem}
 \item $\AbsP((O,H),a,x_b^{\abs}\in I,\Res',(O_{\ret},H_{\ret}))$. Moreover, if $\ret\in H'$ then  $H_{\ret} =H'$ and $O'=O_{\ret}$; otherwise, $p_{\infty}\notin H'$  and $O'$ consists of the live  obligations
  $(\live,\prec $$u)$ such that  $(f,\prec $$u)\in O_{\ret}$ for some $f\in \{\live,\first\}$.
 \end{compactitem}
   \item Case $\gamma'= \bad$: $\call\in H$, $(x_b^{\abs}\in I)\in H$, state $(q,O,H)$ is terminal,   $p_{\infty}\in H$, $p_{\infty}\in H'$, $\ret\notin H'$, $O'=\emptyset$, and $\Res'=\emptyset$.
\end{enumerate}
Note that if $I\neq \{\NULL\}$, then the constraint $x_b^{\abs}\in I$ is checked by the constraint  $\con(I)$ (recall that if $I= \{\NULL\}$, then $\con(I)=\true$). The push transitions of point~1 consider the case  where $\Au_{x^{\abs}_b}$ guesses that the current call position $i_c$ has a matching return $i_r$. In this case, the set of obligations and the check state for the next abstract position $i_r$ along the current \MAP\ are pushed on the stack in order to be recovered at the matching-return $i_r$. Moreover, if $\Au_{x^{\abs}_b}$ guesses that the next position $i_c+1$ is not $i_r$ (i.e., $\ret\notin H'$), then all the upper-bound obligations in $O_{\ret}$ are propagated as  live obligations at the next position $i_c+1$ (note that the \MAP\ visiting $i_c+1$ starts at $i_c+1$, terminates at $i_r-1$, and does not satisfy proposition $p_{\infty}$).
The push transitions of point~2 consider instead the case  where $\Au_{x^{\abs}_b}$ guesses that the current call position $i_c$ has no matching return $i_r$, i.e., $i_c$ is the last position of the current \MAP. In this case, $\Au_{x^{\abs}_b}$ pushes the symbol $\bad$ on the stack and the transition relation is consistently updated.
\paragraph{Internal transitions:} for each internal transition $q\, \der{a,\theta \wedge x_b^{\abs}\in I, \Res}\,q'$ of $\Au$, we add   the internal transitions
$(q,O,H) \, \der{a,\theta\wedge \con(I), \Res\cup\Res'}\,(q',O',H')$ such that $b\in H$ iff $a=b$, and
\begin{enumerate}
  \item Case $\ret\in H'$: $\intA\in H$,  $(x_b^{\abs}\in I)\in H$, state $(q,O,H)$ is terminal, and $\Res' = \emptyset$.
  \item Case $\ret\notin  H'$:  $\AbsP((O,H),a, x_b^{\abs}\in I,\Res',(O',H'))$.
\end{enumerate}
In the first case, $\Au_{x^{\abs}_b}$ guesses that
the current internal position $i$ is the last one of the current \MAP\ ($\ret\in H'$), while in the second case  the current \MAP\ visits the next non-return position $i+1$.
\paragraph{Pop transitions:} for each pop transition $q\, \der{a,\theta \wedge x_b^{\abs}\in I, \Res,\pop(\gamma)}\,q'\in \Delta_r$, we have the pop  transitions
$(q,O,H) \, \der{a,\theta\wedge \con(I), \Res\cup\Res',\pop(\gamma')}\,(q',O',H')$ such that $b\in H$ iff $a=b$, and
\begin{enumerate}
  \item Case $\gamma\neq \top$: $\ret\in H$, $\gamma'=(\gamma,(O,H))$, and  $(x_b^{\abs}\in I)\in H$.
  If $\ret\notin H'$, then  $\AbsP((O,H),a, x_b^{\abs}\in I,\Res',(O',H'))$; otherwise, $(q,O,H)$ is a terminal state and $\Res'=\emptyset$.
\item Case $\gamma = \top$: $\ret\in H$, $I=\{\NULL\}$,  $\gamma'=\top$,  $p_{\infty}\in H$,  $p_{\infty}\in H'$, and $O=\emptyset$.
If $\ret\notin H'$, then  $\AbsP((O,H),a, x_b^{\abs}\in I,\Res',(O',H'))$; otherwise,  $\Res'=\emptyset$ and $O'=\emptyset$.
\end{enumerate}
If $\gamma\neq \top$, then the current return position has a matched-call. Otherwise, the current position is also the initial position
of the associated \MAP.

Finally, the generalized B\"{u}chi condition $\mathcal{F}'$ of $\Au_{x^{\abs}_b}$ is defined as follows. For each B\"{u}chi component $F$ of $\Au$,
  $\Au_{x^{\abs}_b}$  has the B\"{u}chi component consisting of the states $(q,O,H)$ such that $q\in F$.
  Moreover,
  $\Au_{x^{\abs}_b}$ has an additional component consisting of the states $(q,O,H)$ such that $p_{\infty}\in H$.
 Such a component  ensures that  the guesses about the matched calls are correct. Finally,
 for each lower-bound constraint $x^{\abs}_b\succ \ell$ (resp., upper-bound constraint $x^{\abs}_b\prec u$) of $\Au$, $\Au_{x^{\abs}_b}$ has a B\"{u}chi component consisting of the states $(q,O,H)$ such that
 \begin{compactitem}
   \item  $p_{\infty}\in H$, and \emph{either} $b\in H$, \emph{or} the unique constraint in $H$ is equivalent to $x_b^{\abs}\succ \ell \wedge x_b^{\abs} \prec u'$ for some upper-bound
   $u'$, \emph{or} $\succ$$\ell\notin O$;
   \item (resp., $p_{\infty}\in H$, and \emph{either} $b\in H$, \emph{or} the unique constraint in $H$ is equivalent to $x_b^{\abs}\succ \ell' \wedge x_b^{\abs} \prec u$ for some lower-bound
   $\ell'$, \emph{or} $(\first,\prec$$ u)\notin O$).
 \end{compactitem}
 Thus, the above
B\"{u}chi component ensures that
  along an infinite \MAP\ $\nu$, \emph{either} there are infinitely many occurrences of $b$-positions,\emph{ or} there are infinitely many occurrences of positions
  where an interval constraint $x^{\abs}_b \in I$  consistent with $x^{\abs}_b \succ\ell$ (resp., $x^{\abs}_b \prec u$) is used, \emph{or} there are infinitely many positions in $\nu$ where the set of obligations
  does not contain $\succ$$\ell$ (resp., $(\first,\prec$$ u)$).

\section{Removal of caller event-clocks in nested \VPTA}\label{APP:RemovalCallerClocs}

In this section, we prove the following result.

\begin{theorem}[Removal of caller event-clocks] \label{theorem:RemovalCallerClocs}
Given a generalized B\"{u}chi nested  \VPTA\ $\Au$ with set of event clocks $C$   and a caller event-clock $x_b^{\caller}\in C$, one can construct in singly exponential time a generalized B\"{u}chi nested  \VPTA\ $\Au_{x_b^{\caller}}$ with set of  event clocks $C\setminus\{x_b^{\caller}\}$  such that
$\TLang(\Au_{x_b^{\caller}})=\TLang(\Au)$ and $K_{\Au_{x_b^{\caller}}}=K_{\Au}$. Moreover, $\Au_{x_b^{\caller}}$ has $O(n\cdot 2^{O(p)})$ states and $m + O(p)$ clocks, where $n$ is the number of $\Au$-states, $m$ is the number
of standard $\Au$-clocks, and $p$ is the number of \emph{event-clock} atomic constraints on $x_b^{\caller}$ used by $\Au$.
\end{theorem}

Fix a generalized B\"{u}chi nested \VPTA\ $\Au=\tpl{\Sigma, Q,Q_{0},C\cup C_{st},\Gamma\cup\{\top\},\Delta,\mathcal{F}}$ such that $x^{\caller}_b \in C$.
We construct a generalized B\"{u}chi nested \VPTA\ $\Au_{x^{\caller}_b }$ equivalent to $\Au$ whose set of event clocks is  $C\setminus \{x^{\caller}_b\}$, and whose set of standard clocks is
$C_{st}\cup C_\new$, where $C_\new$ consists of the  fresh standard clocks $z_{\succ \ell}$ (resp., $z_{\prec u}$) for each lower-bound  constraint  $x_b^{\caller}\succ \ell$  (resp., upper-bound constraint $x_b^{\caller} \prec u$)
of  $\Au$ involving $x_b^{\caller}$. Since a caller path from a position $j$ consists only of call positions except position $j$ (if $j\notin \Scall$), we assume that $b\in \Scall$
(the case where $b\notin\Scall$ is straightforward).

 The main idea of the construction is that whenever $b$ occurs at a call position $i_c$ of the input word, the simulating automaton $\Au_{x^{\caller}_b}$ guesses the set of lower-bound and upper-bound constraints on
 $x_b^{\caller}$ that will be used by $\Au$ along the \MAP\ $\nu$ having $i_c$ as caller. Note that such a \MAP\ is empty if $i_c+1$ is a return, and starts at position $i_c+1$ otherwise.

  First, let us consider lower-bound constraints $x_b^{\caller} \succ \ell$. Assume that $b$ occurs at a call position $i_c$ of the input word and $i_c+1$ is not a return. Let $\nu$ be the \MAP\ starting at position $i_c+1$. Then, $\Au_{x^{\caller}_b}$ guesses the set of lower-bound constraints  $x_b^{\caller} \succ \ell$  that will be used by $\Au$ along $\nu$. For each of such
    guessed constraints $x_b^{\caller} \succ \ell$, $\Au_{x^{\caller}_b}$ resets the associated new clock $z_{\succ \ell}$, and moves to the next position by carrying in the control state the new set of lower-bound obligations $\succ$$\ell$.
    On visiting the positions $j$ of $\nu$, $\Au_{x^{\caller}_b}$ checks that the guess is correct by verifying that for the current lower-bound constraint
    $x_b^{\caller} \succ \ell'$ used by $\Au$, $\succ$$\ell'$ is in the current set of obligations, and constraint $z_{\succ \ell'}\succ\ell'$ holds. Moreover, at position $j$,
   $\Au_{x^{\caller}_b}$ guesses whether the constraint $x_b^{\caller} \succ \ell'$ will  be again used along   $\nu$, or not. In the first case, the obligation $\succ$$\ell'$ is kept, otherwise, it is discarded.
   If a new call $n_c$ occurs along $\nu$ before the last position of $\nu$, then all the \emph{caller paths} starting from the positions   $h\in [n_c+1,n_r-1]$, where $n_r$ is the matching return of $n_c$ (i.e., $n_r$ is the position following $n_c$ along $\nu$), visit positions $i_c$ and  $n_c$ ($n_c> i_c$).  Thus, the fulfillment of a lower-bound constraint  $x_b^{\caller} \succ \ell$ asserted at a position $h\in [n_c+1,n_r-1]$  always implies the fulfillment of the same constraint when asserted at a position $j\geq i_r$ of $\nu$. Therefore, if $b$ occurs at the new call-position $n_c$,  the clocks $z_{\succ \ell}$
  associated with the guessed lower-bound constraints $x_b^{\caller} \succ \ell$ used by $\Au$ along the \MAP\ having $n_c$ as caller (such a \MAP\ starts at position $n_c+1$ and leads to position $n_r-1$) can be safely reset.

\noindent  Overall, at each position $i$, $\Au_{x^{\caller}_b}$ keeps track in its control state
whether the caller path from $i$ visits a $b$-position preceding $i$, or not. In the first case,
$\Au_{x^{\caller}_b}$ also keeps track in its control state of the
set of obligations associated with the guessed lower-bound constraints on $x^{\caller}_b$  which will be used by $\Au$ in the suffix of the current \MAP\ from position $i$.
In the second case, there are no obligations.
 Whenever  a matched call $i_c\geq i$ occurs along $\nu$, the guessed lower-bound obligations (if any) for the matching return $i_r$ of $i_c$ are pushed on the stack in order to be recovered at position $i_r$. Moreover, if $i_c+1$ is not a return (i.e., $i_r\neq i_c+1$), and either we are in the first case or $i_c$ is a $b$-position,  then $\Au_{x^{\caller}_b}$ guesses the set $L$ of lower-bound constraints which will be used by
 $\Au$ in the finite \MAP\ starting at position $i_c+1$, and moves to the next position by carrying in its control state the obligations associated with $L$.
 Additionally, if $i_c$ is a $b$-position, then for each $x_b^{\caller} \succ \ell\in L$,  the associated new clock $z_{\succ \ell}$ is reset.

The situation for  upper-bound constraints $x_b^{\caller}  \prec u$ is dual. In this case, as in the proof of Theorem~\ref{APP:RemovalAbstractRecorder}, we distinguish between
fresh upper-bound obligations ($\first$,$\prec$$u$) and live upper-bound obligations ($\live$,$\prec$$u$).
Fresh upper-bound obligations $(\first,\prec$$ u)$ always refer to guesses done along the current \MAP\, and differently from the live upper-bound obligations, may be removed, when along the current \MAP, they are checked.
Live upper-bound obligations ($\live$,$\prec$$u$) are propagated from the caller \MAP\ to the called \MAP. They are used by $\Au_{x^{\caller}_b}$ to remember that
at a matched $b$-call position $i_c$ along the current \MAP\ with matching return $i_r>i_c+1$, if the upper-bound constraint $x_b^{\caller}  \prec u$ is guessed to be used by $\Au$ along the finite \MAP\ $\nu'$ having as caller $i_c$ ($\nu'$ starts at $i_c+1$ and ends at $i_r-1$), and the guessed set of obligations for the matching return $i_r$ already contains an obligation ($f$,$\prec$$u$), then the clock  $z_{\prec u}$ must not be reset. This is safe since
the fulfillment of an upper-bound constraint  $x_b^{\caller} \prec u$ asserted at a position   $j\geq i_r$ along $\nu$ always implies the fulfillment  of the same constraint when asserted at a position $h$ of the \MAP\ $\nu'$.

The formal definition of $\Au_{x^{\caller}_b}$ is similar to that of the nested \VPTA\ $\Au_{x^{\abs}_b}$ exploited in the proof of Theorem~\ref{APP:RemovalAbstractRecorder}. Thus, here, we omit the details of the
construction.

\end{document}